\documentclass{article}
\usepackage[latin1]{inputenc}
\usepackage{amsfonts}
\usepackage{amsthm}
\usepackage{amsmath,bm}\interdisplaylinepenalty=2500
\usepackage{amssymb}
\usepackage{nccmath}
\usepackage[english]{babel}
\usepackage{hyperref}
\usepackage[all]{xy}
\usepackage{tikz-cd}
\usepackage{cite}
\usepackage{url}

\newtheorem{lemma}{Lemma}
\newtheorem{example}[lemma]{Example}
\newtheorem{definition}[lemma]{Definition}
\newtheorem{proposition}[lemma]{Proposition}
\newtheorem{theorem}{Theorem}
\newtheorem{remark}[lemma]{Remark}
\newtheorem{notation}[lemma]{Notation}
\newtheorem{corollary}[lemma]{Corollary}

\renewcommand{\_}{\rule{.6em}{.5pt}\hspace{0.023cm}}
\DeclareMathSymbol{:}{\mathbin}{operators}{"3A}

\newcommand{\C}{\mathcal C}
\newcommand{\D}{\mathcal D}
\newcommand{\E}{\mathcal E}
\newcommand{\F}{\mathcal F}
\renewcommand{\r}{\rightarrow}
\newcommand{\Id}{\mathrm{Id}}
\newcommand{\I}{\mathbb{I}}
\newcommand{\End}{\mathrm{End}}
\newcommand{\U}{\mathcal U}
\newcommand{\Hom}{\mathrm{Hom}}
\newcommand{\Cat}{\mathrm{Cat}}
\newcommand{\Lex}{\mathrm{Lex}}
\newcommand{\Clan}{\mathrm{Clan}}
\newcommand{\Mon}{\mathrm{Mon}}
\newcommand{\Gl}{\lambda}
\newcommand{\id}{\mathrm{id}}

\newcommand{\Set}{\mathrm{Set}}
\newcommand{\leftmultimap}{\rotatebox[origin=c]{180}{ $\multimap$ }}

\newcommand{\M}{{\mathcal M}}
\newcommand{\pullback}[1]{\underset{#1}{\times}}

\newcommand{\colim}{\mathrm{colim}}

\title{Notions of parametricity as monoidal models for type theory}
\date{\today}

\author{Hugo Moeneclaey\\
Universit\'e de Paris, \\
Inria Paris, CNRS, IRIF, \\
France\\
}

\begin{document}

\maketitle

\begin{abstract}
This article gives a solid theoretical grounding to the observation that cubical structures arise naturally when working with parametricity. Following \cite{moeneclaey2021parametricity}, we claim that cubical models are cofreely parametric.

We use categories, lex categories or clans as models of type theory. In this context we define notions of parametricity as monoidal models, and parametric models as modules. This covers not only the usual parametricity where any type comes with a relation, but also variants where it comes with a predicate, a reflexive relation, two relations, and many more.

In this setting we prove that forgetful functors from parametric models to arbitrary ones have left and right adjoints. Moreover we give explicit compact descriptions for these freely and cofreely parametric models.

Then we give many examples of notion of parametricity, allowing to build the following as cofreely parametric models:
\begin{itemize}
\item Categories of cubical objects for any variant of cube.
\item Lex categories of truncated semi-cubical (or cubical with reflexivities only) objects.
\item Clans of Reedy fibrant semi-cubical (or cubical with reflexivities only) objects.
\end{itemize}

\end{abstract}

\tableofcontents

\section*{Overview}

\subsection*{Introduction to type theory}

Type theories are foundational systems for constructive mathematics \cite{martin1975intuitionistic}. A type theory is given by its rules, which allow to build types (corresponding to sets and propositions) and terms in types (corresponding to elements in sets and proofs of propositions). 

In this article we adopt a semantical point view \cite{dybjer1995internal,hofmann1997syntax}. This means that we will study models of type theory. Such a model consists of notions of types and terms obeying the rules of the given type theory. From this point of view the syntax of type theory constitutes the simplest model, technically the initial model. 

Many definitions of model have been proposed in the literature, but many of them have two important common features:
\begin{itemize}
\item They are generalised algebraic notions \cite{cartmell1986generalised}. This means that we have a natural notion of morphism between models, forming a locally presentable category \cite{adamek1994locally}. These categories are well-behaved, for example all forgetful functors between locally presentable categories always have left adjoints.
\item A model of type theory has an underlying category. This allows to translate many categorical intuitions and techniques to models of type theory.
\end{itemize}
It is not uncommon to ask for enough structure in a model of type theory so that it can reasonably be considered as a place where all (constructive) mathematics can happen. In this article we will not consider so much structure in models, focusing on modeling types and dependencies between them.

\subsection*{Introduction to parametricity and cubes}

Parametricity was first introduced as a syntactic propriety of system F \cite{reynolds1983types}. It can be summarised as the observation that any type comes with a relation and any term preserves these relations. This is proved by mutual induction on types and terms. 

Intuitively, this means that terms in system F treat their inputs uniformly. For example it implies that any term of type $X\r X$ with $X$ a type variable is extensionally equal to the identity function \cite{wadler1989theorems}. This techniques has been extended from system F to type theory \cite{takeuti2001theory,bernardy2010parametricity,keller2012parametricity}.

We will consider parametric models, which are models with:
\begin{itemize}
\item For each type $A$ a relation $A_*$ over $A$.
\item For each term $t$ a new term $t_*$ witnessing that $t$ preserves the relevant relations.
\end{itemize}
Moreover $A_*$ and $t_*$ should obey the equations inductively defining parametricity. It is a known phenomena that when attempting to build parametric models, cubical structures arise:
\begin{itemize}
\item A model for (a variant of) parametricity was given by Bernardy, Coquand and Moulin using (a variant of) cubical sets \cite{bernardy2015presheaf}.
\item Johann and Sojakova defined so-called \emph{higher-dimensional parametricity} using cubes in \cite{johann2017cubical}, 
\item Cavallo and Harper gave a cubical type theory (i.e. a type theory inspired by cubical models) supporting some kind of parametricity \cite{cavallo2020internal}.
\end{itemize} 
In this article we will explain and illustrate this phenomena.

\subsection*{Content}

Many variants of parametricity can be considered, where any type comes with something else than a relation. In this article we show that for many such variant of parametricity we have a string of adjoint functors: 

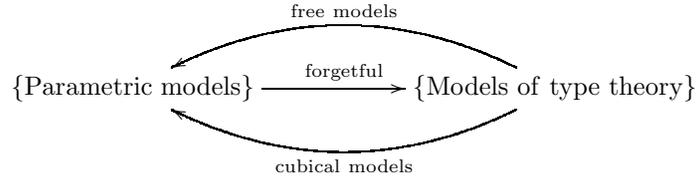
\[\xymatrix{
\{\mathrm{Parametric\ models}\}\ar[rr]^{\mathrm{forgetful}} & & \ar@/^-2pc/[ll]_{\mathrm{free\ models}}\ar@/^2pc/[ll]^{\mathrm{cubical\ models}}\{\mathrm{Models\ of\ type\ theory}\}
}\]

\noindent as in \cite{moeneclaey2021parametricity}. We provide a few notions of parametricity with their corresponding cubical structure, to give a bit of intuition:


\[\begin{tabular}{|c|c|c|}
\hline
Parametricty & Relation &  Semi-cubes\\
\hline
Internal parametricity & Reflexive relations & Cubes\\
\hline
Unary parametricity & Predicate & Augmented semi-simplices\\
\hline
Bi-parametricity & Two relations & Semi-bicubes\\
\hline
\end{tabular}\]

We will consider three notions of model for type theory:

\begin{itemize}
\item Plain categories, without any extra structure. This is usually not considered as a model of type theory. Nevertheless we believe that this example is crucial in providing an introduction to parametricity for categorically-minded readers.
\item Left exact categories (lex categories for short). They are particularly simple models for type theory with $\Sigma$-types and a unit type.
\item Clans \cite{joyal2017notes}, generalising lex categories. They model type theory with $\Sigma$-types and a unit, where a type $A$ over $\Gamma$ is fully determined by its weakening: 
\begin{eqnarray}
(\Gamma,A) &\r& \Gamma
\end{eqnarray} 
A clan is a category with a class of map called fibration obeying a few of the axioms of model categories, so they will not look too unfamiliar to homotopically-minded reader. 
\end{itemize}

We give an early disclaimer on conceivable extensions:

\begin{itemize}
\item 
We do not foresee any trouble when extending the consideration presented here to categories with families, comprehension categories and similar notions of models of type theory.
\item We conjecture that extensions to models with more inductive types are fine, provided that we have enough inductive types to define all the parametricity relations involved.
\item Having $\Pi$-types and universes lead to serious trouble, and we do not attempt to treat them here.
\end{itemize}

Now we give a precise plan:
\begin{itemize}
\item In Section \ref{Categories} we define the usual parametricity (where every type comes with a relation) for categories. Motivated by a reformulation of this definition, we define a notion of parametricity as a monoidal category, and a parametric model as a module. We see that the monoidal category correponding to the usual parametricity is the (opposite of the) category of semi-cubes.
\item In Section \ref{Axiomatisation} we axiomatise the situation from Section \ref{Categories}, by asking for a symmetric monoidal closed category of models. Then we define a notion of parametricity as a monoid in this category (that is a model of type theory with a compatible monoidal structure on its underlying category). Moreover we define parametric models as modules, and we give compact descriptions for the left and right adjoints to the functor forgetting the module structure.
\item In Section \ref{LexCategories} we present notions of parametricity for lex categories as monoids for a symmetric monoidal closed structure on the category of lex categories. We show that notions of parametricity for categories can be extended to lex categories. Moreover we show that non-iterated parametricities can be defined for lex categories, giving rise to truncated semi-cubical (or cubical with reflexivities only) objects, for example graphs. 
\item In Section \ref{Clans} we present notions of parametricity for clans. We show how notions of parametricity for categories can be extended to clans, giving rise to various cubical objects with pointwise fibrations. Moreover we show that clans of semi-cubical (or cubical with reflexivities only) objects with Reedy fibrations are cofreely parametric for suitable notions of parametricity.
\end{itemize}

\section{Parametricity for categories}

\label{Categories}

This preliminary section is meant to motivate the abstract axiomatisation of the next section, by explaining the reasoning leading to it. We consider plain categories as models for type theory, and examine parametricity for them very closely.

\subsection{A first definition}

Recall the most basic form of parametricity is when any type comes with a relation, and any term preserves these relations. First we translate this as directly as possible to a category. 

\begin{definition}
A category $\C$ is parametric if we are given:
\begin{itemize}
\item An endofunctor: 
\begin{eqnarray}
\__* &:& \C\r \C
\end{eqnarray}
\item For any $X:\C$ two morphisms:
\begin{eqnarray}
d^0_X,d^1_X &:& X_* \r X
\end{eqnarray}
natural in $X$.
\end{itemize}
\end{definition}

\begin{remark}
This definition reflects our intuition on parametricity. Indeed:
\begin{itemize} 
\item For $Xa:\C$ the data of:
\begin{eqnarray}
d^0_X,d^1_X &:& X_*\r X
\end{eqnarray}
is a relation over $X$ internal to $\C$.
\item The fact that $\__*$ is a functor with $d^0$ and $d^1$ natural precisely means that morphisms in $\C$ preserve the relations.
\end{itemize}
\end{remark}

Given an object $X$ in a parametric category $\C$, we can iterate $\__*$, building the following diagram:

\[\xymatrix{
X & X_{*}\ar@<-.5ex>[l]_{d^0_X}\ar@<.5ex>[l]^{d^1_X} & X_{**} \ar@<-.5ex>[l]_{d^0_{X_*}}\ar@<.5ex>[l]^{d^1_{X_*}} \ar@/^2.0pc/[l]^{(d^1_X)_*}\ar@/^-2.0pc/[l]_{(d^0_X)_*} & \cdots
}\]

This will turn out to be a semi-cubical objects in $\C$ with $X$ as its object of points. See Remark \ref{semiCubicalObject} for more details.


\begin{remark}
A notion of parametricity where $\__*$ can be applied as many times as we want is called iterated. We will tell more about non-iterated (also called truncated) notions of parametricity later. The most basic example is when $\__*$ can only be applied once, so that $X_{**}$ does not give any new information on $X$.
\end{remark}

\begin{example}
We give examples of parametric categories.
\begin{itemize}
\item We can define: 
\begin{eqnarray}
X_* &=& X\\
d^0_X(x) &=& x\\
d^1_X(x) &=& x
\end{eqnarray} 
Any type comes with its equality relation. This works in any category.
\item If the category has products, we can consider: 
\begin{eqnarray}
X_* &=& X\times X\\
d^0_X(x,y) &=& x\\
d^1_X(x,y) &=& y
\end{eqnarray}
Here any type comes with the trivial (i.e. always true) relation. 
\item If we have an initial object $\bot$ we can define:
\begin{eqnarray}
X_* &=& \bot
\end{eqnarray}
Any type comes with the empty (i.e. always false) relation.
\end{itemize}
In all these examples the maps:
\begin{eqnarray}
 (d^0_X,d^1_X) &:& X_*\r X\times X
 \end{eqnarray}
are monomorphisms, so that their fibers are sub-singletons. Assuming the law of excluded middle, there is no other such notions of parametricity on the category of sets.
\end{example}

\begin{example}
\label{exampleParametricityPath}
We give examples where the maps:
\begin{eqnarray}
 (d^0_X,d^1_X) &:& X_*\r X\times X
 \end{eqnarray}
 are not monomorphisms. 
\begin{itemize} 
\item In any cartesian closed category of spaces containing the unit interval $[0,1]$ in $\mathbb{R}$, we can define:
\begin{eqnarray}
X_* &=& X^{[0,1]}\\
d^0_X(p) &=& p(0)\\
d^1_X(p) &=& p(1)
\end{eqnarray}
For any points $x$ and $y$ in $X$, the fiber $X_*$ over $(x,y)$ is the space of paths from $x$ to $y$ in $X$. 

The cubical space constructed from any space $X$ by iterating $\__*$ is the usual cubical nerve of $X$ with its natural topology.
\item This can be generalised to any cartesian closed category with an object $I$ and two maps:
\begin{eqnarray}
0,1&:&\top\r I
\end{eqnarray}
by defining:
\begin{eqnarray}
X_* &=& X^I\\
d^0_X &=& X^I\overset{X^0}\r X^\top\cong X\\
d^1_X &=& X^I\overset{X^1}\r X^\top\cong X
\end{eqnarray}
\item Assuming a functorial factorisation system in a category $\C$, we can define a parametricity for $\C$ by factoring diagonals as follows:
\[\xymatrix{
X\ar[r] & X_*\ar[rr]^{(d^0_X,d^1_X)} & & X\times X
}\]
\end{itemize}
\end{example}

\subsection{Reformulation using semi-cubes}

Recall the following description for category of semi-cubes:

\begin{definition}
\label{definitionSquare}
Let $\square$ is the free strict monoidal category generated by:
\begin{itemize}
\item An object $\I$.
\item Two morphisms:
\begin{eqnarray}
d^0,d^1 &:& \I\r 1
\end{eqnarray}
where $1$ is the monoidal unit.
\end{itemize}
\end{definition}

The reader unfamiliar with semi-cubical objects can take the following as a definition.

\begin{proposition}
\label{squareMonoidal}
A functor from $\square$ to $\C$ is a semi-cubical object in $\C$.
\end{proposition}

\begin{remark}
It is standard to define semi-cubical objects as presheaves on $\square^{op}$, but we do not want to include opposites in our axiomatisation of the situation. We will give remarks whenever this could lead to confusion.
\end{remark}

\begin{remark}
Objects in $\square$ are of the form $\I\otimes\cdots\otimes\I$ and morphisms are tensors of $d^0$, $d^1$ and $\id_\I$. 

For example morphisms going from $\I\otimes\I\otimes\I$ to $\I\otimes\I$ are:
\begin{eqnarray}
d^\epsilon\otimes \I\otimes\I\\
\I\otimes d^\epsilon\otimes \I\\
\I\otimes\I\otimes d^\epsilon
\end{eqnarray}
where $\epsilon =0,1$.
\end{remark}

There is a clear analogy between the natural transformations:
\begin{eqnarray}
d^0,d^1 &:& \__*\r \Id
\end{eqnarray}
and the generators for semi-cubes:
\begin{eqnarray}
d^0,d^1 &:& \I\r 1
\end{eqnarray}

To make this precise we need auxiliary definitions. The first one is very well-known:

\begin{definition}
For $\C$ a category, we write $\End_\C$ for the strict monoidal category of endofunctors of $\C$, with composition as tensor and identity as unit. 
\end{definition}

The second definition is not as common, but elementary.

\begin{definition}
Let $\M$ be a strict monoidal category. An $\M$-module consists of:
\begin{itemize}
\item A category $\C$.
\item A strict monoidal functor from $\M$ to $\End_\C$.
\end{itemize}
\end{definition}

This is just the usual definition of an action of monoid, valid in any cartesian closed category, specialised to the category of categories. 

\begin{remark}
We could equivalently define an $\M$-module as a category $\C$ with a functor: 
\begin{eqnarray}
\_\otimes\_ &:& \M\times \C \r \C
\end{eqnarray}
such that:
\begin{eqnarray}
(i\otimes j)\otimes X &=& i\otimes(j\otimes X)\\
1\otimes X &=& X
\end{eqnarray}
functorially in $i,j$ and $X$. 
\end{remark}

Now we can state:

\begin{lemma}
Being parametric is the same as being a $\square$-module.
\end{lemma}
\begin{proof}
By Proposition \ref{squareMonoidal}, giving a monoidal functor from $\square$ to $\D$ is equivalent to giving an object $X:\D$ with two morphisms from $X$ to $1$. So a monoidal functor from $\square$ to $\End_\C$ corresponds precisely to $\__*$ an endofunctor of $\C$ with two natural transformations from $\__*$ to the identity endofunctor. 
\end{proof}

\begin{remark}
\label{semiCubicalObject}
For any $\square$-module $\C$ with an object $X$, we have a functor:
\begin{eqnarray}
F &:& \square \r \C\\
F(i) &=& i\otimes X
\end{eqnarray}
such that $F(1) = X$, so that $X$ is indeed the object of points for a semi-cubical object in $\C$.
\end{remark}

\section{An axiomatisation for this situation}

\label{Axiomatisation}

Now, using insight from the first section, we will define an abstract notion of parametricity. We will use the example of categories throughout to give intuitions on definitions and results presented here.

\begin{notation}
Everywhere in this section we will assume $\U$ a symmetric monoidal closed category, with:
\begin{itemize}
\item A product $\_\otimes\_$.
\item A unit $1$.
\item A linear arrow $\_\multimap\_$.
\end{itemize}
\end{notation}

An object in $\U$ is seen as a model of type theory.

\begin{example}
We can take the category of categories as $\U$. Its monoidal closed structure is in fact even cartesian.
\end{example}

\begin{example}
We can take the category of abelian groups as $\U$. The tensor is genuinely linear in the sense that it is not a cartesian product.
\end{example}

\subsection{Notions of parametricity and parametric objects}

\begin{definition}
A notion of parametricity for $\U$ is a monoid in $\U$. 

This means it consists of $\M:\U$ with:
\begin{eqnarray}
\mu &:& \M\otimes \M\r \M\\
\epsilon &:& 1\r \M
\end{eqnarray}
such that the following diagrams commute: 
\[\xymatrix{
(\M\otimes \M)\otimes \M \ar[dd]_{\cong} \ar[rr]^{\mu\otimes\M} && \M\otimes\M\ar[rd]^\mu\\
& & &\M\\
\M\otimes(\M\otimes\M) \ar[rr]_{\M\otimes\mu} && \M\otimes\M\ar[ru]_\mu\\
}\]
\[\xymatrix{
\M\ar[r]^\id\ar[d]_\cong & \M \\
\M\otimes 1 \ar[r]_{\M\otimes\epsilon} & \M\otimes\M\ar[u]_{\mu}\\
}\xymatrix{
\M\ar[r]^\id\ar[d]_\cong & \M \\
1\otimes \M \ar[r]_{\epsilon\otimes\M} & \M\otimes\M\ar[u]_{\mu}\\
}\]
\end{definition}

\begin{remark}
Here the diagrams are required to commute up to equality. So if $\U$ is the cartesian closed category of categories, a monoid in $\U$ is a \emph{strict} monoidal category.
\end{remark}


\begin{example}
\label{CatCartesianClosed}
For $\U$ the category of categories, a notion of parametricity is a (strict) monoidal category. We have many examples besides $\square$ from Definition~\ref{definitionSquare}. For example, all categories of cubes from \cite{buchholtz2017varieties} are monoidal. So we have notions of parametricity for categories corresponding to all kinds of cubes:
\begin{itemize}
\item We can have symmetries, diagonals or reflexivities. For reader familiar with cubical type theories, they correspond to structural rules on interval variables.
\item We can have connections and inverses.
\end{itemize}
All notions of parametricity in Example \ref{exampleParametricityPath} have reflexivities. Many more variants are possible, for example:
\begin{itemize}
\item We can consider the free monoidal category with an object $\I$ and $n$ maps to~$1$. This gives parametricity with $n$-ary relations rather than binary ones. For $n=1$ we get unary parametricity, which is a form of iterated realisability.
\item We can consider multiple generating objects in the monoidal category, for example bicubes are generated by two interval objects.
\end{itemize}
\end{example}

\begin{remark}
The notion of unary parametricity is similar to Kleene-style realisability, in the sense that any type comes with a predicate. But they have significant differences:
\begin{itemize}
\item In realisability a formula $A$ is send to a predicate $A_*$ on some programs, and a proof of $A$ is send to program $p$ realising $A$ (i.e. such that $A_*(p)$ holds). For example $A_*$ can be a predicate on untyped $\lambda$-terms, or G\"odel numbers for recursive functions. This means that we can have a realiser for $A$ without any proof of $A$. 
\item In unary parametricity a type $A$ is sent to a predicate $A_*$ on the type $A$ itself, and a term $a:A$ is send to $a_*:A_*(a)$. This is sometimes summarised by saying that type theory is its own language of realisers. Since $A_*$ is predicate over $A$, we cannot consider a realiser for a type without inhabitant.
\end{itemize}
Overall parametricity and realisability have distinct goals:
\begin{itemize}
\item Realisability validates new formulas using computational justifications, i.e. shows formulas consistent by finding programs realising them.
\item Unary parametricity emphasises the fact that terms are similar to programs, i.e. that they are continuous in some sense: they preserve the relevant predicates.
\end{itemize}
\end{remark}

\begin{remark}
The monoidal category corresponding to unary parametricity is freely generated by an object $\I$ and a map:
\begin{eqnarray}
d &:& \I\r 1
\end{eqnarray} 
It is the (opposite of the) category of augmented semi-simplices. This means that augmented semi-simplices are unary semi-cubes!

There is a similar result for the category of augmented simplices, which is freely generated by a comonoid $\I$.
\end{remark}

\begin{remark}
Recall that we used a non-standard orientation for the category of semi-cubes, so we should consider the opposites of the standard cube categories as notions of parametricity. This is not an issue because the opposite of a monoidal category is monoidal.
\end{remark}

Now we can easily adapt our definition of parametric categories.

\begin{remark}
For $\C:\U$ we have a monoid $\End_\C$ in $\U$ with $\C\multimap\C$ as underlying object, composition as product and identity as unit. 
\end{remark}

\begin{definition}
An object $\C:\U$ is called $\M$-parametric if it is an $\M$-module, i.e. if we are given a morphism of monoid from $\M$ to $\End_\C$. 
\end{definition}

A morphism from $\M$ to $\End_\C$ can be reformulated as a map $\alpha:\M\otimes \C\r \C$ such that the following diagrams commute:

\[\xymatrix{
(\M\otimes \M)\otimes \C \ar[dd]_{\cong} \ar[rr]^{\mu\otimes\C} && \M\otimes\C\ar[rd]^\alpha\\
& & &\M\\
\M\otimes(\M\otimes\C) \ar[rr]_{\M\otimes\alpha} && \M\otimes\C\ar[ru]_\alpha\\
}\]

\[\xymatrix{
\C\ar[r]^\id_\C\ar[d]_\cong & \C \\
1\otimes \C \ar[r]_{\epsilon\otimes\C} & \M\otimes\C\ar[u]_{\alpha}\\
}\]

\begin{remark}
In the literature an $\M$-module is sometimes called an $\M$-action, depending on the context. We will sometimes say that $\M$ acts on $\C$.
\end{remark}

Maps between $\M$-modules respecting the $\M$-action are called \emph{equivariant}. We give a precise definition:

\begin{definition}
An equivariant map between $\M$-modules $(\C,\alpha)$ and $(\D,\beta)$ is a map:
\begin{eqnarray}
f&:&\C\r \D
\end{eqnarray}
such that the following square commutes:
\[\xymatrix{
\M\otimes\C\ar[r]^{\M\otimes f}\ar[d]_\alpha & \M\otimes\D\ar[d]^\beta \\
\C\ar[r]_f & \D\\
}\]
\end{definition}

There is  a category of modules and equivariant maps.

\subsection{Freely and cofreely parametric objects}

Let $\M$ be a monoid in $\U$. Now we prove a theorem which will be used extensively in the rest of the paper. It can be summarised by saying that freely and cofreely $\M$-parametric objects always exist.


\begin{theorem}
The forgetful functor from $\M$-modules to $\U$ has both left and right adjoints. The left (resp. right) adjoint sends $\C$ to $\M\otimes \C$ (resp. $\M\multimap \C$) with the action induced by the canonical left (resp. right) action of $\M$ on itself.
\end{theorem}
\begin{proof}
We prove this result in linear simply-typed $\lambda$-calculus, so that it is true in any monoidal closed category. In order to check that we are indeed reasoning linearly, it is crucial that any bound variable in our $\Gl$-terms occurs precisely once. An alternative direct proof by diagram chasing in the symmetric monoidal closed category $\U$ is certainly possible.

We proceed with the proof for the right adjoint. Given $\C:\U$, we define $R(\C)$ as $\M\multimap \C$ with the $\M$-action
 $\alpha$ defined by: 
\begin{eqnarray}
\alpha &:& \M\multimap R(\C) \multimap R(\C)\\
\alpha(i,u) &=& \Gl j.\,u(j\otimes i)
\end{eqnarray}
Which is indeed an action. Now for $f:\C\r \D$ we define: 
\begin{eqnarray}
R(f) &:& R(\C)\multimap R(\D)\\
R(f,u) &=& f \circ u
\end{eqnarray}
We see that $R(f)$ is equivariant, and that $R$ is a functor from $\C$ to $\M$-modules. 

Now we want to check that it is right adjoint to the forgetful functor. Assume $(\C,\alpha)$ an $\M$-action, and $\D:\U$. We define:
\begin{eqnarray}
\psi &:& \Hom_\U(\C,\D) \r \Hom_\U(\C,\M\multimap \D)\\
\psi(f) &=& \Gl c,i.\, f(\alpha(i,c))
\end{eqnarray}
and check that $\psi(f)$ is equivariant. Next we define:
\begin{eqnarray}
\phi &:& \Hom_\U(\C,\M\multimap \D)\r \Hom_\U(\C,\D)\\
\phi(g) &=& \Gl c.\, g(c,1)
\end{eqnarray}
We check that for all $f:\Hom_\U(\C,\D)$ we have:
\begin{eqnarray}
\phi(\psi(f)) &=& \Gl c.\,(\Gl c',i.\, f(\alpha(i,c')))(c,1)\\
&=& \Gl c.\,f(\alpha(1,c))\\
&=& f
\end{eqnarray}
and that for all equivariant $g:\Hom_\U(\C,\M\multimap \D))$, i.e. all $g$ such that: 
\begin{eqnarray}
g(\alpha(i,c)) &=& \Gl j.\, g(c,j\otimes i)
\end{eqnarray} 
we have:
\begin{eqnarray}
\psi(\phi(g)) &=& \Gl c,i.\, (\Gl c'.\, g(c',1))(\alpha(i,c))\\
&=& \Gl c,i.\, g(\alpha(i,c),1)\\
&=& \Gl c,i.\, g(c,1\otimes i)\\
&=& g
\end{eqnarray}
Now we just need naturality to conclude, so we check that for: 
\begin{eqnarray}
\C,\C',\D,\D' &:& \U \\
\alpha &:& \M\ \mathrm{acting}\ \mathrm{on}\ \C\\
\alpha' &:& \M\ \mathrm{acting}\ \mathrm{on}\ \C'\\
f&:&\Hom_\U(\C',\C)\ \mathrm{with}\ f\ \mathrm{equivariant}\\
g&:&\Hom_\U(\C,\D)\\
h&:&\Hom_\U(\D,\D')
\end{eqnarray}
we have that:
\begin{eqnarray}
\psi(h\circ g\circ f) &=& \Gl c,i.\,(h\circ g)(f(\alpha'(i,c)))\\
&=& \Gl c,i.\,(h\circ g)(\alpha(i,f(c)))\\
&=& \Gl c.\, R(h)(\Gl i.\, g(\alpha(i,f(c))))\\
&=& R(h)\circ\psi(g)\circ f 
\end{eqnarray}

We omit the similar proof for the left adjoint.
\end{proof}

\begin{example}
For $\U$ the category of categories and $\square$ the category of semi-cubes:
\begin{itemize} 
\item The cofreely parametric category generated by $\C$ is the category of semi-cubical objects in $\C$, i.e. the category of functor from $\square$ to $\C$. 
\item The freely parametric category generated by $\C$ is the category $\square\times\C$. The existence of this left adjoint is easy to prove, but this formula is still pleasantly explicit. 
\end{itemize}
It works the same for all the previously mentioned variants of cubes, including bicubes and augmented simplices.
\end{example}


\begin{remark}
\label{analogyAbelian}
Recall that $\U$ was simply assumed symmetric monoidal closed. Taking the category of abelian groups as $\U$, we get the following correspondence: 
\[\begin{tabular}{|c|c|}
\hline
Model of type theory & Abelian groups \\
Notions of parametricity & Rings \\
Parametric models & Modules \\
Freely parametric models & Free modules \\
Cubical models & Cofree modules\\
\hline
\end{tabular}\]
Free (resp. cofree) modules are often called induced (resp. coinduced) modules.
\end{remark}



\section{Parametricity for left exact categories}

\label{LexCategories}

In this section we apply the machinery from Section \ref{Axiomatisation} to left exact categories (abbreviated as lex categories), i.e. categories with finite limits. 
So we will prove that the category of lex categories is symmetric monoidal closed, so that we can define notions of parametricity as monoids and parametric lex categories as modules. Then we will give examples of such monoids. It should be noted that we use a very strict version of lex categories:
\begin{itemize}
\item A lex category is given with a chosen limit for any finite diagram.
\item Functors between lex categories have to preserve the chosen limits.
\item Limits commute strictly, rather than up to natural isomorphisms.
\end{itemize}
These restrictions allow us to successfully use the $1$-category of lex categories. For example it is immediate to see that we have adjoint functors freely adding and forgetting finite limits. 

\begin{remark}
Already with categories in Example \ref{CatCartesianClosed}, we defined a notion of parametricity as a \emph{strict} monoidal category, bringing some unnecessary strictness. 

An alternative approach would be to use a symmetric monoidal closed $2$-category of models of type theory, with a notion of parametricity defined as a monoid where equations hold up to $2$-isomorphisms. 
\end{remark}

\begin{remark}
Most notions of parametricity considered here are finitely presented, so that assuming them strict is painless.
\end{remark}

\subsection{The category of lex categories is symmetric monoidal closed}

First we give our unusually strict definition of a lex category. We will require that some morphisms are identities, implicitly requiring that their sources and targets are equal. 

\begin{definition}
A lex category is a category $\C$ with an operation $\lim_I$ for any finite category $I$ sending any $F:\C^I$ to a limit cone.

Moreover we ask that for any finite categories $I$ and $J$ and diagram $H:\C^{I\times J}$, we have that the canonical isomorphisms:
\begin{eqnarray}
\lim_{i:I}(\lim_{j:J} H(i,j)) &\r& \lim_{j:J}(\lim_{i:I} H(i,j))
\end{eqnarray}
are identities.
\end{definition}

We define morphisms of lex categories as functor preserving the chosen limit cones. This gives a category $\Lex$. It should be noted that this defines $\Lex$ as the category of algebras for an essentially algebraic extension of the theory of categories. In particular this means that the functor forgetting limits has a left adjoint freely adding them. 

We define a symmetric monoidal closed structure on the category of lex categories. The arrow is immediate:

\begin{definition}
For $\C$ and $\D$ lex categories we define $\C\multimap\D$ as the category of lex functors from $\C$ to $\D$ with:
\begin{itemize}
\item Natural transformations as morphisms.
\item Limits computed pointwise.
\end{itemize} 
\end{definition}

\begin{remark}
This definition would not be valid without the strict commutations of limits. Indeed for $(f_i)_{i:I}$ a finite diagram of lex functors, its pointwise limit being lex means that for any diagram $(x_j)_{j:J}$ we have that the canonical morphism:
\begin{eqnarray}
\lim_{i:I}\lim_{j:J}f_i(x_j) &\r& \lim_{j:J}\lim_{i:I}f_i(x_j)
\end{eqnarray}
is an identity. 

This phenomena is already present for groups: the pointwise product of two group morphisms is not necessarily a group morphism, when the target group is not abelian.
\end{remark}

The tensor is defined so that it will be left adjoint to the arrow.

\begin{definition}
For $\C$ and $\D$ lex categories, we define $\C\otimes\D$ as the free lex category generated by a functor:
\begin{eqnarray}
\_\otimes\_ &:& \C\times\D\r \C\otimes\D
\end{eqnarray}
such that the canonical morphisms:
\begin{eqnarray}
(\lim_{i:I}c_i)\otimes d &\r& \lim_{i:I}(c_i\otimes d)\\
c \otimes (\lim_{j:J}d_j) &\r& \lim_{j:J}(c\otimes d_j)
\end{eqnarray}
are identities.
\end{definition}

This means that in order to define a a lex functor from $\C\otimes\D$ to $\E$, it is enough to define a functor $F:\C\times\D\r \E$ such that the morphisms:
\begin{eqnarray}
F(\lim_{i:I}c_i,d)&\r& \lim_{i:I}F(c_i,d)\\
F(c,\lim_{j:J}d_j)&\r& \lim_{j:J}F(c,d_j) 
\end{eqnarray}
are identities.

\begin{remark}
This definition implies that the canonical morphism:
\begin{eqnarray}
\lim_{i:I}\lim_{j:J} (c_i\otimes d_j) &\r& \lim_{j:J}\lim_{i:I} (c_i\otimes d_j)
\end{eqnarray}
is an identity. This equality would be unnatural, although not incoherent, without the strict commutativity of limits.
\end{remark}


\begin{definition}
We define $1$ as the free lex category generated by an object.
\end{definition}

So giving a lex functor from $1$ to $\C$ is the same as giving an object in $\C$.

\begin{remark}
The category with finite colimits freely generated by an object is the category of finite sets, so that $1$ is equivalent to the opposite of the category of finite sets.
\end{remark}

\begin{remark}
\label{analogy1}
This is very similar to the tensor and arrow for abelian groups, with the following correspondence extending Remark \ref{analogyAbelian}.
\[\begin{tabular}{|c|c|}
\hline
Sets & Categories\\
Addition, zero & Finite limits\\
Abelian groups & Lex categories\\
\hline
\end{tabular}\]
We will make this formal in Remark \ref{commutativeMonads}.
\end{remark}

We write $U$ for the forgetful functor sending a lex category to its underlying category. We write $F$ for its left adjoint freely adding limits. We prove an auxiliary lemma.

\begin{lemma}
\label{enrichedAdjunction}
For any category $\C$ and lex category $\D$ we have natural isomorphisms:
\begin{eqnarray}
U(F(\C)\multimap\D) &\cong& U(\D)^\C
\end{eqnarray}
Moreover limits in $F(\C)\multimap\D$ correspond to pointwise limits in $U(\D)^\C$.
\end{lemma}
\begin{proof}
The isomorphism is immediate on objects. For any:
\begin{eqnarray}
X&:&\Hom_\Lex(F(\C),\D)
\end{eqnarray}
we write $\overline{X}$ the corresponding functor from $\C$ to $U(\D)$.

For $\D$ a category, we denote by $\D^\r$ its arrow category with $S$ (resp. $T$) the functor giving the source (resp. the target) of an arrow. Assume $X$ and $Y$ in $\Hom_\Lex(\C,\D)$, then a natural transformation from $X$ to $Y$ is:
 \begin{eqnarray}
\{ Z : \Hom_\Cat(\C,\D^\r)\ |\ S\circ Z = X,\ T\circ Z = G\}
\end{eqnarray}
But using the fact that $X$ and $Y$ are lex and naturality, we can show that any such $Z$ is in fact lex. Then for any $X$ and $Y$ in $F(\C)\multimap\D$ we have:
  \begin{eqnarray}
&& \Hom_{F(\C)\multimap\D}(X,Y)\\
&\cong&\{ Z : \Hom_\Cat(F(\C),\D^\r)\ |\ S\circ Z = X,\ T\circ Z = Y\} \\
&\cong & \{ Z : \Hom_\Lex(F(\C),\D^\r)\ |\ S\circ Z = X,\ T\circ Z = Y\} \\
&\cong & \{ \overline{Z} : \Hom_\Lex(\C,U(\D)^\r)\ |\ S\circ \overline{Z} = \overline{X},\ T\circ \overline{Z} = \overline{Y}\} \\
&\cong& \Hom_{U(\D)^\C}(\overline{X},\overline{Y})
\end{eqnarray}
This concludes the proof that:
\begin{eqnarray}
U(F(\C)\multimap\D) &\cong& U(\D)^\C
\end{eqnarray}

This isomorphism restricts functor and natural transformations defined on $F(\C)$ to $\C$, and limits in $F(\C)\multimap\D$ are computed pointwise, so they correspond to pointwise limits in $U(\D)^\C$.
\end{proof}

\begin{remark}
In principle a natural transformation between lex functors should be assumed lex. This condition can be omitted (and always is) because it is always true, as used in the above proof.
\end{remark}

Now we bring all these constructions together.

\begin{theorem}
\label{monoidalLexCat}
This defines a symmetric monoidal closed structure on the category of lex categories.
\end{theorem}
\begin{proof}
There are many things to check:
\begin{itemize}
\item First we check that we have a natural isomorphism:
\begin{eqnarray}
\Hom_\Lex(\C\otimes\D,\E) \cong \Hom_\Lex(\C,\D\multimap\E)
\end{eqnarray}
Indeed giving a morphism in: 
\begin{eqnarray}
\Hom_\Lex(\C\otimes\D,\E)
\end{eqnarray} 
is naturally equivalent to giving:
\begin{eqnarray}
F &:& \Hom_\Cat(\C\times\D,\E)
\end{eqnarray}
such that the canonical morphisms:
\begin{eqnarray}
F(\lim_{i:I}c_i,d) &\r& \lim_{i:I}F(c_i,d)\\
F(c,\lim_{j:J}d_j)&\r&\lim_{j:J}F(c,d_j)
\end{eqnarray}
are identities, which is in turn naturally equivalent to giving:
\begin{eqnarray}
\overline{F} &:& \Hom_\Cat(\C,\E^\D)
\end{eqnarray}
such that the canonical morphisms:
\begin{eqnarray}
\overline{F}(\lim_{i:I}c_i)(d) &\r& \lim_{i:I}\overline{F}(c_i)(d)\\
\overline{F}(c)(\lim_{j:J}d_j)&\r&\lim_{j:J}\overline{F}(c)(d_j)
\end{eqnarray}
are identities. The first conditions means that $\overline{F}$ is a lex functor from $\C$ to $\E^\D$, with limits in $\E^\D$ computed pointwise. The second condition means that the image of $\overline{F}$ is included in $\D\multimap\E$, i.e. the full subcategory of $\E^\D$ consisting of lex functors. So together they precisely mean that:
\begin{eqnarray}
\overline{F} &:& \Hom_\Lex(\C,\E\multimap\D)
\end{eqnarray}
\item Next we check that $\_\otimes\_$ is symmetric. Indeed we can check that the the functor: 
\begin{eqnarray}
S &:& \C\times\D\r \D\otimes\C\\
S(c,d) &=& d\otimes c
\end{eqnarray}
commutes with limits in $c$ and $d$, so that indeed it can be extended to:
\begin{eqnarray}
S &:& \Hom_\Lex(\C\otimes\D,\D\otimes\C)
\end{eqnarray}
We can check that $S$ is self-inverse. 
\item Similarly we can define a functor:
\begin{eqnarray}
A &:& (\C\times\D)\times\E \r \C\otimes(\D\otimes\E)\\
A((c,d),e) &=& c\otimes(d\otimes e)
\end{eqnarray}
which commutes with limits in each variable, so that it can be extended to:
\begin{eqnarray}
A &:& \Hom_\Lex((\C\otimes\D)\otimes \E,\C\otimes(\D\otimes\E))\\
\end{eqnarray}
It is straightforward to define an inverse to $A$.
\item Now we need to check that $1$ is indeed a unit. But this is a consequence of the natural isomorphism:
\begin{eqnarray}
(1\multimap\C) &\cong& \C
\end{eqnarray}
in $\Lex$, by Lemma \ref{enrichedAdjunction} applied to $1=F(\top)$. 
\end{itemize}
We omit the necessary checking of the various coherence diagrams.
\end{proof}

\begin{remark}
\label{commutativeMonads}
There exists a notion of commutative monad (see for example Section 6 in \cite{brandenburg2014tensor}). For $T$ a commutative monad on a symmetric monoidal closed category $\C$, the category of $T$-algebras is symmetric monoidal closed (assuming equalisers in $\C$ to build arrows and coequalisers in $T$-algebras to build tensors). We give two examples.
\begin{itemize}
\item The monad of abelian groups on sets is commutative.
\item The monad of lex categories on categories is commutative, with our assumption: 
\begin{eqnarray}
\lim_I\lim_J &=& \lim_J\lim_I
\end{eqnarray} 
Otherwise it would be commutative in a $2$-categorical sense.
\end{itemize}
The monoidal closed structures on the categories of abelian groups and lex categories can both be built this way, cementing the analogy between abelian groups and lex categories from Remark \ref{analogy1}.
\end{remark}

\subsection{Notions of parametricity for lex categories}

Now we can use the symmetric monoidal closed structure from the previous section to define notions of parametricity as well as parametric models. We emphasise this:

\begin{definition}
A notion of parametricity for lex categories is a monoid in $\Lex$.
\end{definition}

Now we get the following proposition by unfolding this definition:

\begin{proposition}
Giving a notion of parametricity for lex categories is equivalent to giving a lex category $\M$ with a (strictly) monoidal product $\_\otimes\_$ such that the canonical morphisms:
\begin{eqnarray}
(\lim_{i:I}m_i)\otimes n &\r& \lim_{i:I}(m_i\otimes n)\\
m \otimes (\lim_{j:J}n_j) &\r& \lim_{j:J}(m\otimes n_j)
\end{eqnarray}
are identities.
\end{proposition}

Such a category will be called a monoidal lex category. Now we give examples. First we prove that notions of parametricity for categories can be extended to lex categories.

\begin{proposition}
\label{freeLexMonoidal}
The functor $F$ freely adding finite limits to a category is strongly monoidal.
\end{proposition}
\begin{proof}
We write $U$ for the functor forgetting finite limits. 
\begin{itemize}
\item We have a string of natural isomorphisms where $\C$ and $\D$ are categories and $\E$ is a lex category:
\begin{eqnarray}
\Hom_\Lex(F(\C\times \D),\E) &\cong& \Hom_\Cat(\C\times\D,U(\E))\\
&\cong& \Hom_\Cat(\C,U(\E)^\D)\\
&\cong& \Hom_\Cat(\C,U(F(\D)\multimap\E))\label{keyStep}\\
&\cong& \Hom_\Lex(F(\C),F(\D)\multimap\E)\\
&\cong&\Hom_\Lex(F(\C)\otimes F(\D),\E)
\end{eqnarray} 
where Equation \ref{keyStep} uses Lemma \ref{enrichedAdjunction}. We can conclude by Yoneda lemma.
\item The lex category $1$ is actually defined as $F(\top)$ with $\top$ the terminal category. 
\end{itemize}
We omit the necessary checking of the various coherence diagrams.
\end{proof}


\begin{corollary}
For any notion of parametricity $\M$ for categories, we have that $F(\M)$ is a notion of parametricity for lex categories.
\end{corollary}
\begin{proof}
Strongly monoidal functors preserve monoids.
\end{proof}


Now that we are considering lex categories, we can deal with truncated notions of parametricity, i.e. non-iterated ones. We fix $n$ a natural number for this section. First we define the $n$-truncated semi-cubes:

\begin{definition}
Let $\square$ be the semi-cube category. Any object in $\square$ if of the from $\I^k = \I\otimes\cdots \otimes\I$ the $k$-th monoidal product of $\I$. We write $\square_n$ the full subcategory of $\square$ with objects $1,\I,\cdots,\I^n$.
\end{definition}

Our goal is to show that $\square_n$ induces a notion of parametricity for lex categories. To do this we need to prove that the free lex category generated by $\square_n$ is a monoidal lex category.

Recall that the Day convolution extends $\_\otimes\_$ on $\square$ to a monoidal tensor on $\Set^\square$. This tensor is closed on both sides, meaning that we have two bifunctors $\_\multimap\_$ and $\_\leftmultimap\_$ with natural isomorphisms:
\begin{eqnarray}
\Hom(X\otimes Y,Z) &\cong& \Hom(X,Y\multimap Z)\\
\Hom(X\otimes Y,Z) &\cong& \Hom(Y,Z \leftmultimap X)
\end{eqnarray}
for in $X,Y,Z:\Set^\square$.

The inclusion of full subcategory $f:\square_n\r \square$ induces a post-composition functor:
\begin{eqnarray}
f^*  &:& \Set^\square \r \Set^{\square_n}
\end{eqnarray}
with a full and faithful left (resp. right) adjoint $f_!$ (resp. $f_*$). An object $X:\Set^{\square}$ is called coskeletal if: 
\begin{eqnarray}
X &\cong &f_*(f^*(X))
\end{eqnarray}
and $f_*(f^*(X))$ is called the coskeleton of $X$.

\begin{remark}
Coskeletal cubical objects are usually called $n$-coskeletal, to emphasise the dependency on $n$. We omit $n$ here because it is fixed for the whole section.
\end{remark}

We want to give a helpful criteria for coskeletal object. First we define $n$-cells:

\begin{definition}
\label{borderInductiveDefinition}
We define:
\begin{eqnarray}
\delta\I^k &:& \Set^\square\\
\delta i^k &:& \delta\I^k\r \I^k
\end{eqnarray}
by:
\begin{eqnarray}
\delta i^0 &:& \bot \r 1
\end{eqnarray}
and we define $\delta i^{k+1}$ from $\delta i^k$ using the pushout square:
\[\xymatrix{
 & \delta \I^k \coprod \delta \I^k\ar[dl]_{\alpha_{\delta\I^k}}\ar[dr]^{\delta i^k \coprod \delta i^k}& \\
 \delta \I^k\otimes \I\ar[rd]\ar@/^-2pc/[rdd]_{\delta i^k\otimes \I} &  &\I^n\coprod\I^n\ar[ld]\ar@/^2pc/[ldd]^{\alpha_{\I^k}}\\
  & \delta\I^{k+1}\ar[d]^{\delta i^{k+1}} & \\
  & \I^{k+1}& \\
}\]
where:
\begin{eqnarray}
\alpha_X = (X\otimes d^0\, |\, X\otimes d^1) :  &:& X\coprod X \r X\otimes\I
\end{eqnarray}
\end{definition}

One can check that the map:
\begin{eqnarray}
\delta i^k &:& \delta\I^k\r\I^k
\end{eqnarray}
is the inclusion of the border of a $k$-cell in the usual sense.

\begin{definition}
An morphism $u:A\r B$ is left orthogonal to an object $X$ if the induced map:
\begin{eqnarray}
u^* &:& \Hom(B,X) \r \Hom(A,X)
\end{eqnarray}
is a bijection. We denote this by $u\bot X$.
\end{definition}

This means that for any $f:A\r X$ there exists a unique dotted arrow making the triangle commutes:
\[\xymatrix{
A\ar[d]_u\ar[r]^f & X\\
B\ar@{-->}[ru] & \\
}\]

We do not give a proof for the next lemma. It holds both for semi-cubes and cubes with refexivities only (as claimed without proof in \cite{kennett2011levels}), although we do not know to what extent it holds for other cubes.

\begin{lemma}
An element $X:\Set^\square$ is coskeletal if and only $\delta i^k$ for $k>n$ is left orthogonal to $X$.
\end{lemma}

\begin{lemma}
\label{lemmaCoskeleton}
Assume $X$ coskeletal, then so are $\I\multimap X$ and $X\leftmultimap \I$.
\end{lemma}
\begin{proof}
\begin{itemize}
\item We proceed with the proof for that $\I\multimap X$ is coskeletal when $X$ is. We need to prove that:
\begin{eqnarray} 
\delta i^k\ \bot\ (\I\multimap X)
\end{eqnarray}
for $k>n$. This is equivalent to:
\begin{eqnarray} 
(\delta i^k\otimes \I)\ \bot\ X
\end{eqnarray}
We can decompose $\delta i^k\otimes \I$ as:
\[\xymatrix{
\delta \I^k\otimes \I\ar[r]^\theta& \delta \I^{k+1} \ar[r]^{\delta i^{k+1}}& \I^{k+1}
}\]
where $\theta$ is a pushout of $\delta i^k \coprod \delta i^k$. But since maps left orthogonal to $X$ are stable by coproducts, pushouts and compositions, we can conclude.

\item To give the proof for $X\leftmultimap \I$ we need to rework this section in mirror, with an alternative equivalent definition of $\delta \I^{k+1}$ based on $\I\otimes \delta \I^{k}$ rather than $\delta \I^{k}\otimes\I$.
\end{itemize}
\end{proof}

Now we are ready to restrict the Day product from $\Set^\square$ to $\Set^{\square_n}$

\begin{lemma}
The category $\Set^{\square_n}$ inherits a monoidal structure from the Day convolution on $\Set^\square$. This induced tensor on $\Set^{\square_n}$ commutes with colimits in both variables.
\end{lemma}
\begin{proof}
We proceed in three steps.
\begin{itemize}
\item If $Y$ is coskeletal, so are $\I\multimap Y$ and $Y \leftmultimap \I$. This is Lemma \ref{lemmaCoskeleton}, and this is the only part which rely on the kind of cube we use.

\item If $Y$ is coskeletal, so are $X\multimap Y$ and $Y \leftmultimap X$. By iterating the previous point, the property is true when $X$ is representable, and it is stable by colimits because coskeletal objects are stable by limits as $f_*$ and $f^*$ preserve limits.
\item For any $X,Y:\Set^\square$ we have natural isomorphisms:
\begin{eqnarray}
f^*(f_!f^*(X)\otimes Y) &\cong& f^*(X\otimes Y) \\
f^*(X\otimes f_!f^*(Y)) &\cong& f^*(X\otimes Y) 
\end{eqnarray}
indeed we have for example:
\begin{eqnarray}
\Hom(f^*(f_!f^*(X)\otimes Y),Z) & \cong & \Hom(f_!f^*(X)\otimes Y,f_*(Z))\\
&\cong& \Hom(f_!f^*(X),Y\multimap f_*(Z))\\
&\cong& \Hom(X,f_*f^*(Y\multimap f_*(Z)))\\
&\cong& \Hom(X,Y\multimap f_*(Z)) \label{isoCoskeletal}\\
&\cong& \Hom(X\otimes Y,f_*(Z))\\
&\cong& \Hom(f^*(X\otimes Y),Z)
\end{eqnarray}
where Equation \ref{isoCoskeletal} used the fact that $f_*(Z)$ is coskeletal, so that: 
\begin{eqnarray}
Y\multimap f_*(Z)
\end{eqnarray}
is coskeletal as well by the previous point. 

\item For $X,Y:\Set^{\square_n}$, we can define:
\begin{eqnarray} 
X\otimes_n Y &=& f^*(f_!(X)\otimes f_!(Y))\\
1_n &=& f^*(1)
\end{eqnarray}
Using the previous point we can check that this gives a monoidal structure.
\end{itemize}
The functor $\_\otimes_n\_$ commutes with colimits in both variables because so does $\_\otimes\_$, and $f^*$ and $f_!$ commutes with colimits.
\end{proof}

This proposition does depend on $\square$ when we check that if $Y$ is coskeletal, so are $\I\multimap Y$ and $Y \leftmultimap \I$. This is true for semi-cubes and cubes with reflexivities only, but we do not how far this can be extended.

\begin{proposition}
The free lex category generated by $\square_n$ is (non-strict) monoidal.
\end{proposition}
\begin{proof}
By duality, the previous lemma gives a monoidal structure $\_\otimes_n\_$ on $(\Set^{\square_n})^{op}$, commuting with limits in both variables. 

The free lex category generated by $\square_{n}$ is equivalent to the closure of representables in $(\Set^{\square_n})^{op}$ by finite limits. We want to restrict $\_\otimes_n\_$ to this closure. Since $\_\otimes_n\_$ commutes with finite colimits in both variables, it is sufficient to prove that $\I^k\otimes_n\I^{k'}$ is a finite colimit of representable in $\Set^{\square_n}$ for $k,k'\leq n$ in order to conclude. We will prove the more general fact that $f^*(\I^l)$ is a finite limit of $\I^k$ with $k\leq n$ for all $l$.

We know that for any object $X:\Set^{\square_n}$ we have:
\begin{eqnarray}
X &\cong& \colim_{\, i : \square_n^{op}, X(i)}\, \Hom_\square(i,\_)
\end{eqnarray}
so that $X$ is a colimit of representable. This is called the co-Yoneda lemma.

If $X=f^*(\I^l)$, then $X(i)$ is finite (as $X(i) = \Hom_\square(l,i)$ and $\square$ is locally finite), so that $X$ is a finite limit of representables.
\end{proof}

Here we used the fact that $\square$ was locally finite (i.e. has finite hom-sets) in a crucial way.

\begin{remark}
The monoidal category from the previous proposition is not strict, so that technically it is not a notion of parametricity. We expect that this can be worked around using a $2$-category of models of type theory, or alternatively a strictification result for lex monoidal categories.
\end{remark}

A $\square_{n}$-parametric lex category has an endofunctor $\__*$ such that for any $X$, the object $X_{*(n+1)}$ can be computed as a limit from $X_{*k}$ with $k\leq n$. A cofreely $\square_{n}$-parametric lex category is simply a category of $n$-cubical objects in some lex category. 

\begin{remark}
In the usual type-theoretic point of view on parametricity, the condition on $\__*$ can be reformulated simply as $X_{*(n+1)}=\top$, as the limit involved is taken care of by the complicated context in which $X_{*(n+1)}$ is defined.
\end{remark}

We conjecture that this $n$-truncated parametricity for lex categories actually works for all kinds of cubes.




\section{Parametricity for clans}

\label{Clans}

The notion of clans was designed by Joyal to help bridge the gap between type theory and homotopy theory \cite{joyal2017notes}. It is based on the idea that types can be modeled by fibrations as follows:
\begin{itemize}
\item A fibration $f:A\r \Gamma$ is a map with fibers varying continuously in $\Gamma$.
\item A dependent type over $\Gamma$ is family of types varying continuously in $\Gamma$.
\end{itemize}
So a dependent type over a context $\Gamma$ is modeled in a clan by a fibration with target $\Gamma$. 

Axiomatising the notion of fibration has been done in many different ways in the past, notably in model categories \cite{quillen1967axiomatic} and categories of fibrant objects \cite{brown1973abstract}. Fibrations in clans retain very little features from these classical homotopical axiomatisations. 

There is no major new idea compared to the last section, but everything is technically more involved. First we will give our unusually strict definition of clans and show that the category of clans is symmetric monoidal closed. Then we will give examples of notions of parametricity for clans and examine the corresponding cubical models.

\subsection{Clans with strictly commuting limits}

Here we give our definition of clans. We define the usual clans and call them weak clans, so that we can define clans as weak clans with strictly commuting limits. First we give an auxiliary definition.

\begin{definition}
A class of maps $\F$ is called closed by pullbacks if for any: 
\begin{eqnarray}
f:c\r b\\
g : a\r b
\end{eqnarray}
with $f$ in $\F$, there exists a pullback square:
\[\xymatrix{
a\pullback{b} c \ar[d]_{\pi_a}\ar[r]^{\pi_c} & c\ar[d]^f\\
a \ar[r]_{g}& b
}\]
where $\pi_a$ is in $\F$. For any commutative square:
\[\xymatrix{
d\ar[d]_{p}\ar[r]^{q} & c\ar[d]^f\\
a \ar[r]_{g}& b
}\]
we write:
\begin{eqnarray}
\langle p,q\rangle &:& d\r a\pullback{b} c
\end{eqnarray}
the unique map such that:
\begin{eqnarray}
\pi_a\circ\langle p,q\rangle &=& p\\
\pi_c\circ\langle p,q\rangle &=& q
\end{eqnarray}
\end{definition}

Recall that clans attempt to capture the intuition of a class of maps with continuously varying fibers.

\begin{definition}
A weak clan is a category with a terminal object $\top$, together with a class of maps called fibrations such that:
\begin{itemize}
\item Fibrations are closed by isomorphisms, compositions and pullbacks.
\item Maps to $\top$ are fibrations.
\end{itemize} 
\end{definition}

We write:
\begin{eqnarray}
f: a &\twoheadrightarrow&b
\end{eqnarray}
to say that $f$ is a fibration. 


\begin{remark}
The map $\id_a$ is isomorphic to the pullback of $\id_\top$ along the unique map from $a$ to $\top$, so that it is a fibration. 
\end{remark}

\begin{remark}
Cartesian products are defined in any weak clan as pullbacks:
\[\xymatrix{
a\times b\ar[r]^{\pi_b}\ar[d]_{\pi_a} & b\ar[d]\\
a \ar[r]& \top
}\]
So cartesian projections are fibrations. They correspond to non-dependent types.
\end{remark}

\begin{remark}
We have a correspondence between weak clans and type theory, where a type over a context $\Gamma$ is a fibration with target $\Gamma$. We explain assumptions on fibrations from this point of view:
\[\begin{tabular}{|c|c|c|}
\hline
Fibrations & Types \\
\hline
Pullbacks & Substitutions\\
Identities & $\top$ \\
Compositions & $\Sigma$\\ 
Maps to $\top$ & Democracy\\ 
Projections & Non-dependent types\\ 
\hline 
\end{tabular}\]
\end{remark}

Now we want to define strictly commuting limits for clans. We recall some vocabulary:

\begin{definition}
A square:
\[\xymatrix{
d\ar@{->>}[r]\ar@{->>}[d] & c\ar@{->>}[d]\\
a\ar@{->>}[r] & b\\
}\]
is called Reedy fibrant if the induced map:
\begin{eqnarray}
d &\r& a\pullback{b}c 
\end{eqnarray}
is a fibration.
\end{definition}

Now we can define clans.

\begin{definition}
A clan is a weak clan with strictly commuting limits. This means that:
\begin{itemize}
\item The canonical morphism:
\begin{eqnarray}
\top \r \top\pullback{\top}\top
\end{eqnarray}
is the identity.
\item For any diagram:
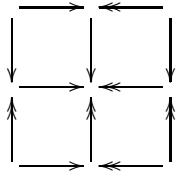
\[\xymatrix{
 \*\ar[r]\ar[d]&  \*\ar[d] & \*\ar@{->>}[l]\ar[d]\\
   \* \ar[r]&  \* &  \*\ar@{->>}[l]\\
    \* \ar[r]\ar@{->>}[u]&  \*\ar@{->>}[u] &  \*\ar@{->>}[l]\ar@{->>}[u]\\
}\]
where the bottom right square is Reedy fibrant, the canonical morphism between its limit computed row-by-row or column-by-column is the identity.
\end{itemize} 
\end{definition}
\begin{proof}
To check that this is meaningful, we have to check that for any diagram:
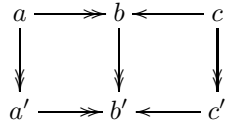
\[\xymatrix{
    a \ar@{->>}[r]\ar@{->>}[d] &  b\ar@{->>}[d] &  c\ar[l]\ar@{->>}[d]\\
   a' \ar@{->>}[r]&  b' &  c'\ar[l]\\
  }\]
with the left square Reedy fibrant, we have that the induced map:
\begin{eqnarray}
a\pullback{b}c &\r& a'\pullback{b'}c' 
\end{eqnarray}
is a fibration. Indeed the following diagram shows that it is isomorphic to the composite of pullbacks of fibrations. 
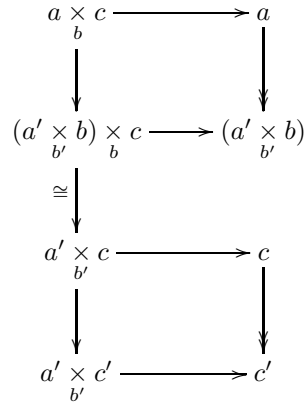
\[\xymatrix{
  a\pullback{b} c \ar[r]\ar[d]& a\ar@{->>}[d]\\
  (a'\pullback{b'}b)\pullback{b} c\ar[d]_{\cong} \ar[r]& (a'\pullback{b'}b)\\
a'\pullback {b'}c\ar[r]\ar[d]& c\ar@{->>}[d] \\
  a'\pullback{b'}c'\ar[r] & c'
}\]
\end{proof}

We will see how strictly commutative limits ensure the existence of a linear arrow for clans. As for lex categories this restriction is necessary because we consider morphisms of clans commuting with limits on the nose. We could presumably use the 2-category of weak clans, without further assumptions on limits. Now we give an auxiliary lemma on free clans.

\begin{lemma}
\label{freeClans}
Fibrations in the free clan are maps isomorphic to projections. 
\end{lemma}
\begin{proof}
It is straightforward to see that any clan has cartesian products, and that maps isomorphic to cartesian projections are always fibrations. Now we prove that the class of maps isomorphic to a projection is a valid class of fibrations:
\begin{itemize}
\item They are stable by isomorphisms by definition.
\item They are stable by compositions as for any objects $a,b,c$ we have:
\[\xymatrix{
(a\times b)\times c \ar[r]^{\pi_{a\times b}}\ar[d]_\cong & a\times b \ar[r]^{\pi_a} & a\ar[d]^\id\\
a\times (b\times c) \ar[rr]_{\pi_\Gamma} & & a
}\]
\item They are stable by pullbacks because for any objects $a,b,c$ and morphism $\sigma:c\r a$, the square:
\[\xymatrix{
c\times b\ar[r]^{\sigma\times b}\ar[d]_{\pi_c} & a\times b\ar[d]^{\pi_a} \\
c\ar[r]_\sigma & a
}\]
is a pullback square.
\item Maps to $\top$ are isomorphic to projections:
\[\xymatrix{
a \ar[r]^{\cong} & \top\times a \ar[r]^{\pi_\top} & \top
}\]
\end{itemize}
So fibrations are precisely maps isomorphic to projections.
\end{proof}

\subsection{Monoidal structure on clans}

Now we define the tensor product, the linear arrow and the unit of clans.

\begin{definition}
Let $\C$ and $\D$ be clans. We define a clan $\C\multimap\D$ where:
\begin{itemize}
\item Objects are morphisms of clans.
\item Morphisms are natural transformations between the underlying functors.
\item Limits are computed pointwise.
\item A natural transformation $\alpha : F\r G$ is a fibration if for any fibration $c\twoheadrightarrow c'$ in $\C$ we have an induced fibration in $\D$:
\begin{eqnarray}
F(c) \twoheadrightarrow F(c')\pullback{G(c')} G(c)
\end{eqnarray}
\end{itemize}
\end{definition}

\begin{proof}
We need to check that this is indeed a clan. As a preliminary result, we prove that fibrations are pointwise fibrations. Indeed for: 
\begin{eqnarray}
\alpha&:&F\twoheadrightarrow G
\end{eqnarray} 
a fibration and $c$ in $\C$, from the fact that $c\twoheadrightarrow\top$ we have a fibration:
\begin{eqnarray}
H(c) &\twoheadrightarrow& H(\top)\pullback{G(\top)}G(c)
\end{eqnarray}
but this map is isomorphic to the map $\alpha(c) : H(c)\r G(c)$. Now we check that limits are well defined:

\begin{itemize} 

\item The constant functor $F_\top$ with value $\top$ is a morphism of clan, giving a terminal object in $\C\multimap\D$. The main thing to check is that the induced map:
\begin{eqnarray}
F_\top(a\pullback{b}c)&\r&F_\top(a)\pullback{F_\top(b)}F_\top(c)
\end{eqnarray}
is an identity but this map is equal to the unique map in:
\begin{eqnarray}
\top&\r&\top\pullback{\top}\top
\end{eqnarray}
which is an identity.

\item Now we check that pullbacks can be defined pointwise in $\C\multimap\D$. So we assume given a diagram in $\C\multimap\D$ as follows:
\[\xymatrix{
F\ar[r] & G & H\ar@{->>}[l]
}\]
And we define $F\pullback{G}H$ by:
\begin{eqnarray}
(F\pullback{G}H)(c) &=& F(c)\pullback{G(c)}H(c)
\end{eqnarray}
We check that it commutes with limits. We need to check that the morphism:
\begin{eqnarray}
(F\pullback{G}H)(\top) &\r& \top
\end{eqnarray}
is an identity so that $F\pullback{G}H$ preserves the terminal object. But this map is the inverse to:
\begin{eqnarray}
\top&\r&\top\pullback{\top}\top
\end{eqnarray}
which is assumed to be an identity.

To check that $F\pullback{G}H$ commutes with pullbacks, we consider a diagram:
\[\xymatrix{
a\ar[r] & b & c\ar@{->>}[l]
}\]
in $\C$. We need to prove that the canonical morphism from the limit of the diagram:
\[\xymatrix{
F(a)\ar[r] \ar[d]& F(b)\ar[d] & F(c)\ar@{->>}[l]\ar[d]\\
G(a)\ar[r] & G(b) & G(c)\ar@{->>}[l]\\
H(a)\ar[r]\ar@{->>}[u] & H(b)\ar@{->>}[u] & H(c)\ar@{->>}[l]\ar@{->>}[u]
}\]
computed row-by-row and column-by-column is an identity. This is a special case of strict commutation of limits because the bottom left square is Reedy fibrant by the definition of the fibration $H\twoheadrightarrow G$.

\end{itemize}

Next we check that the class of fibrations is suitably closed.

\begin{itemize}

\item Stability by isomorphisms is straightforward.

\item Compositions of fibrations are fibrations. Indeed for a diagram:
\[\xymatrix{
F\ar@{->>}[r] & G \ar@{->>}[r]& H
}\]
we need to prove that for any fibration $c\twoheadrightarrow c'$ in $\C$ the map:
\begin{eqnarray}
F(c)\r F(c')\pullback{H(c')}H(c)
\end{eqnarray}
is a fibration in $\D$. But this map is a composite of fibrations as follows:
\[\xymatrix{
F(c) \ar@{->>}[d] & \\
F(c')\pullback{G(c')}G(c) \ar[d]\ar[r] & G(c)\ar@{->>}[d] \\ 
F(c')\pullback{G(c')} (G(c')\pullback{H(c')}H(c)) \ar@{->>}[d]^{\cong}\ar[r] & G(c')\pullback{H(c')}H(c) \\
F(c')\pullback{H(c')}H(c)&  \\
}\]

\item Pullbacks of fibrations are fibrations. We have already checked that they exist and are computed pointwise. Now we need to check that given a diagram:
\[\xymatrix{
F\ar[r] & G & H\ar@{->>}[l]
}\]
in $\C\multimap\D$ the map:
\begin{eqnarray}
F\pullback{G}H\r F
\end{eqnarray}
is a fibration, so that we need to prove that for any fibration $c\twoheadrightarrow c'$ in $\C$, the induced map:
\begin{eqnarray}
(F\pullback{G}H)(c)\r F(c) \pullback{F(c')} (F\pullback{G}H)(c')
\end{eqnarray}
is a fibration. But this map is isomorphic to the map:
\begin{eqnarray}
F(c)\pullback{G(c)}H(c) \r F(c)\pullback{G(c')} H(c')
\end{eqnarray}
 which is isomorphic to a pullback of fibration as follows:
 \[\xymatrix{
 F(c)\pullback{G(c)}H(c) \ar[d]\ar[r] & H(c)\ar@{->>}[d]\\
F(c)\pullback{G(c)}(G(c)\pullback{G(c')}H(c'))\ar[r] \ar[d]^{\cong} & G(c)\pullback{G(c')}H(c')\\
F(c)\pullback{G(c')} H(c') & \\
 }\]
 
\item Maps to $\top$ are fibrations, indeed the unique map:
\begin{eqnarray}
F &\r&\top
\end{eqnarray}
is a fibration if for any fibration: 
\begin{eqnarray}
c\twoheadrightarrow c'
\end{eqnarray} 
in $\C$ the map:
\begin{eqnarray}
F(c) \r F(c')\pullback{\top}\top
\end{eqnarray}
is a fibration in $\D$. But this map is isomorphic to the map:
 \begin{eqnarray}
F(c) &\r& F(c')
\end{eqnarray}
which is a fibration because $F$ preserves fibrations.

\end{itemize}

Finally we can check that limits are strictly commuting because they are computed pointwise.
\end{proof}

Now we define the tensor product of two clans.

\begin{definition}
Given clans $\C$ and $\D$, we define $\C\otimes\D$ as the free clan generated by:
\begin{itemize}
\item A functor:
\begin{eqnarray}
\_\otimes\_ : \C\times\D\r \C\otimes\D 
\end{eqnarray}
\item The fact that for any two fibrations:
\begin{eqnarray} 
c&\twoheadrightarrow& c'\\
d&\twoheadrightarrow& d'
\end{eqnarray} 
we have a fibration:
\begin{eqnarray}
c\otimes d &\twoheadrightarrow& (c'\otimes d)\pullback{c'\otimes d'}(c\otimes d')
\end{eqnarray}
\item For any object $c:\C$ and $d:\D$ the canonical morphisms:
\begin{eqnarray}
c\otimes\top &\r& \top\\
\top\otimes d &\r& \top
\end{eqnarray}
are identities.
\item Given a span:
\[\xymatrix{c_1\ar[r] & c_2 & c_3\ar@{->>}[l]}\]
in $\C$ and a span:
\[\xymatrix{d_1\ar[r] & d_2 & d_3\ar@{->>}[l]}\]
in $\D$ with $c:\C$ and $d:\D$, the canonical morphisms:
\begin{eqnarray}
(c_1\pullback{c_2}c_3)\otimes d &\r& (c_1\otimes d)\pullback{c_2\otimes d}(c_3\otimes d)\\
c\otimes(d_1\pullback{d_2}d_3)&\r& (c\otimes d_1)\pullback{c\otimes d_2}(c\otimes d_3)
\end{eqnarray}
are identities.
\end{itemize}
\end{definition}


\begin{proof}
To check that the last axiom makes sense, we need to show that for any fibration $c\twoheadrightarrow c'$ in $\C$ and $d$ and object in $\D$ we have a fibration:
\begin{eqnarray}
c\otimes d &\twoheadrightarrow& c'\otimes d
\end{eqnarray}
and the same with the role of $\C$ and $\D$ reversed. This is true because:
\begin{itemize}
\item By the second assumption in $\C\otimes\D$ applied to $c\twoheadrightarrow c'$ and $d\twoheadrightarrow\top$ we have:
\begin{eqnarray}
c\otimes d &\twoheadrightarrow& (c'\otimes d)\pullback{c'\otimes\top} (c\otimes\top)
\end{eqnarray}
\item By the third assumption, this fibration is isomorphic to:
\begin{eqnarray}
c\otimes d &\twoheadrightarrow& c'\otimes d
\end{eqnarray}
\end{itemize}
\end{proof}

\begin{remark}
In homotopy theory the usual condition on a tensor is dual, i.e. it is a condition of stability for cofibrations, involving a pushout rather than a pullback. This is not surprising because we use the opposite from the standard categories of cubes. 
\end{remark}

Finally we define the unit.

\begin{definition}
We define $1$ as the free clan generated by an object.
\end{definition}

\begin{remark}
The clan $1$ is the opposite to the category of finite sets, with monomorphisms as fibrations.
\end{remark}

And now we are ready to prove that this defines a monoidal structure. We will proceed exactly as for lex categories, with a few more properties to check in order to take fibrations into account. We denote by $F$ the left adjoint to the forgetful functor $U$ from clans to categories. 

\begin{lemma}
\label{enrichedAdjunctionClans}
For any category $\C$ and clan $\D$ we have a natural isomorphism:
\begin{eqnarray}
U(F(\C)\multimap\D) &\cong& U(\D)^\C
\end{eqnarray}
Moreover:
\begin{itemize} 
\item Fibrations in $F(\C)\multimap\D$ correspond to natural transformations in $U(\D)^\C$ which are pointwise fibrations.
\item Limits in  $F(\C)\multimap\D$ correspond to pointwise limits in $U(\D)^\C$.
\end{itemize}
\end{lemma}
\begin{proof}
The natural isomorphism is very similar to the proof for lex categories in Lemma \ref{enrichedAdjunction}, using the clan structure on $\D^\r$ with the fibrations defined pointwise. The new result we need to prove is that for any two morphisms of clans $X,Y:F(\C)\multimap\D$ with:
\begin{eqnarray}
Z&:&\Hom_\Cat(F(\C),\D^\r)\\
S\circ Z &=& X\\
T\circ Z &=& Y
\end{eqnarray}
we have that $Z$ sends fibrations to pointwise fibrations in $\D^\r$. This means that for any fibration $i : c\twoheadrightarrow c'$ in $F(\C)$ we have that the vertical maps in:
\[\xymatrix{
X(c)\ar[d]_{X(i)}\ar[r]^{Z(c)} & Y(c)\ar[d]^{Y(i)}\\
X(c') \ar[r]_{Z(c')}& Y(c')
}\]

are fibrations. This is true because $X$ and $Y$ preserves fibrations.

Now we want to check that fibrations in $F(\C)\multimap\D$ corresponds to pointwise fibrations in $U(\D)^\C$. Assume given:
\begin{eqnarray} 
X, Y &:&F(\C)\multimap\D\\
\alpha &:& \Hom_{F(\C)\multimap\D}(X,Y)
\end{eqnarray}
with:
\begin{eqnarray} 
\overline{X}, \overline{Y} &:&U(\D)^\C\\
\overline{\alpha} &:& \Hom_{U(\D)^\C}(\overline{X},\overline{Y})
\end{eqnarray}
the corresponding elements in $U(\D)^\C$. The following are equivalent: 

\begin{itemize}
\item The morphism $\alpha$ is a fibration.
\item For all $c\twoheadrightarrow c'$ in $F(\C)$, we have an induced fibration:
\begin{eqnarray}
X(c) &\twoheadrightarrow& X(c')\pullback{Y(c')}Y(c)
\end{eqnarray}
\item For all $c$ and $c'$ in $F(\C)$, we have an induced fibration:
\begin{eqnarray}
X(c)\times X(c') &\twoheadrightarrow& X(c)\times Y(c')
\end{eqnarray}
(because any fibration in $F(\C)$ is isomorphic to a projection $c\times c'\r c$ by Lemma \ref{freeClans} and we have:
\[\xymatrix{
X(c\times c')\ar[r]\ar[d]_\cong & X(c)\pullback{Y(c)}Y(c\times c')\ar[d]^\cong\\
X(c)\times X(c')\ar[r] & X(c)\times Y(c')
}\]
so the top arrow is a fibration if and only if the bottom one is).
\item For all $c'$ in $F(\C)$ we have an induced fibration:
\begin{eqnarray}
X(c') &\twoheadrightarrow& Y(c') 
\end{eqnarray}
(because $X(c')\twoheadrightarrow Y(c')$ implies that for any $d:\D$ we have a fibration:
\begin{eqnarray}
d\times X(c') &\twoheadrightarrow& d\times Y(c')
\end{eqnarray}
and we consider $d=X(c)$).
\item For all $c$ in $\C$ we have an induced fibration:
\begin{eqnarray}
\overline{X}(c) &\twoheadrightarrow& \overline{Y}(c) 
\end{eqnarray}
(because $X$ and $Y$ are the extensions of $\overline{X}$ and $\overline{Y}$ commuting with cartesian products, so the condition implies that $X(c)\twoheadrightarrow Y(c)$ for all $c=c_1\times\cdots\times c_n$ and any object in $F(\C)$ is of this form).
\item The morphism $\overline{\alpha}$ is a pointwise fibration.
\end{itemize}

The fact that limits are computed pointwise is straightforward.
\end{proof}

\begin{theorem}
This defines a symmetric monoidal closed structure on the category of clans.
\end{theorem}
\begin{proof}
The proof is an extension of the proof of Theorem \ref{monoidalLexCat} from lex categories to clans, where we have to check at every step that it can be adapted to fibrations.
\begin{itemize}
\item Concerning the isomorphisms:
\begin{eqnarray}
\Hom_\Clan(\C\otimes\D,\E) &\cong& \Hom_\Clan(\C,\D\multimap\E)
\end{eqnarray}
We know from the lex case that giving a limit preserving functor:
\begin{eqnarray}
F &:& \Hom_\Cat(\C\otimes\D,\E) 
\end{eqnarray}
is equivalent to giving a limit preserving:
\begin{eqnarray}
\overline{F} &:& \Hom_\Cat(\C,\D\multimap\E) 
\end{eqnarray}
so it is enough to check that $F$ preserves fibrations if and only if so does $\overline{F}$.

\subitem We have that $F$ is preserves fibrations if for all fibration $c\twoheadrightarrow c'$ in $\C$ and $d\twoheadrightarrow d'$ in $\D$ we have a fibration:
\begin{eqnarray}
F(c\otimes d) &\twoheadrightarrow& F((c'\otimes d)\pullback{c'\otimes d'}(c\otimes d'))\\
&=& F(c'\otimes d)\pullback{F(c'\otimes d')}F(c\otimes d')
\end{eqnarray}
in $\E$.
\subitem On the other hand $\overline{F}$ preserves fibrations if and only if for any fibration $c\twoheadrightarrow c'$ in $\C$ we have a fibration:
\begin{eqnarray}
\overline{F}(c) &\twoheadrightarrow& \overline{F}(c')
\end{eqnarray}
in $\D\multimap \E$, which means precisely that for all fibration $d\twoheadrightarrow d'$ we have a fibration:
\begin{eqnarray}
\overline{F}(c,d) &\twoheadrightarrow& \overline{F}(c,d')\pullback{\overline{F}(c',d')}\overline{F}(c',d)
\end{eqnarray}
so we see that both conditions are equivalent.

\item Next we check symmetry. To extend the result for lex categories, we need to check that:
\begin{eqnarray} 
S&:&\C\otimes\D\r\D\otimes\C
\end{eqnarray} 
preserves fibrations, meaning that for fibrations $c\twoheadrightarrow c'$ in $\C$ and $d\twoheadrightarrow d'$ in $\D$ we have an induced fibration:
\begin{eqnarray}
S(c\otimes d) &\twoheadrightarrow& S((c'\otimes d)\pullback{c'\otimes d'}(c\otimes d'))
\end{eqnarray}
but this map is equal to:
\begin{eqnarray}
d\otimes c &\twoheadrightarrow& (d\otimes c')\pullback{d'\otimes c'}(d'\otimes c) \\
&\cong& (d'\otimes c)\pullback{d'\otimes c'} (d\otimes c')
\end{eqnarray}
which is a fibration in $\D\otimes\C$.

\item Then we check associativity. To extend the result for lex categories, we need to check that:
\begin{eqnarray}
A&:&(\C\otimes\D)\otimes\E \r \C\otimes(\D\otimes\E)
\end{eqnarray} 
preserves fibrations, i.e. that for any fibrations $c\twoheadrightarrow c'$ in $\C$, with $d\twoheadrightarrow d'$ in $\D$ and $e\twoheadrightarrow e'$ in $\E$, we have an induced fibration:
\begin{eqnarray}
A((c\otimes d)\otimes e) &\twoheadrightarrow& A\big(((c\otimes d)'\otimes e)'\big)
\end{eqnarray}
where we use the informal notation:
\begin{eqnarray}
(x\otimes y)' &=& x'\otimes y\pullback{x'\otimes y'}x\otimes y'
\end{eqnarray}
but the maps:
\begin{eqnarray}
(c\otimes d)\otimes e &\twoheadrightarrow& ((c\otimes d)'\otimes e)'\\
c\otimes (d\otimes e) &\twoheadrightarrow& \big(c\otimes (d\otimes e)'\big)'
\end{eqnarray}
are isomorphic as their right-hand sides are isomorphic to the limits of the diagrams:
\[\xymatrix{
(c'\otimes d)\otimes e \ar[dr]\ar[drr]  & (c\otimes d')\otimes e \ar[dr]\ar[dl]& (c\otimes d)\otimes e'\ar[dl]\ar[dll]\\
(c\otimes d') \otimes e' & (c'\otimes d)\otimes e' & (c'\otimes d')\otimes e\\
}\]
and:
\[\xymatrix{
c'\otimes (d\otimes e) \ar[dr]\ar[drr]  & c\otimes (d'\otimes e) \ar[dr]\ar[dl]& c\otimes (d\otimes e')\ar[dl]\ar[dll]\\
c\otimes (d' \otimes e') & c'\otimes (d\otimes e') & c'\otimes (d'\otimes e)\\
}\]

\item Finally we check the unit isomorphisms using the fact that:
\begin{eqnarray}
(1\multimap\C) &\cong& \C 
\end{eqnarray}
by Lemma \ref{enrichedAdjunctionClans} applied to $1=F(\top)$ with $\top$ the terminal category.
\end{itemize}
\end{proof}

\subsection{Notions of parametricity for clans}

Now we can apply our machinery to define a notion of parametricity for clans as a monoidal clan. Explicitly we have:

\begin{corollary}
A notion of parametricity for clans is a clan $\C$ with a strictly monoidal product $\_\otimes\_$ such that:
\begin{itemize}
\item For any two fibrations:
\begin{eqnarray} 
c&\twoheadrightarrow& c'\\
d&\twoheadrightarrow& d'
\end{eqnarray} 
we have an induced fibration:
\begin{eqnarray}
c\otimes d &\twoheadrightarrow& (c'\otimes d)\pullback{c'\otimes d'}(c\otimes d')
\end{eqnarray}
\item For any object $c:\C$ the canonical morphisms:
\begin{eqnarray}
c\otimes\top &\r& \top\\
\top\otimes c &\r& \top
\end{eqnarray}
are identities.
\item Given a span:
\[\xymatrix{c_1\ar[r] & c_2 & c_3\ar@{->>}[l]}\]
in $\C$ with $d:\C$, the canonical morphisms:
\begin{eqnarray}
(c_1\pullback{c_2}c_3)\otimes d &\r& (c_1\otimes d)\pullback{c_2\otimes d}(c_3\otimes d)\\
d\otimes(c_1\pullback{c_2}c_3)&\r& (d\otimes c_1)\pullback{d\otimes c_2}(d\otimes c_3)
\end{eqnarray}
are identities.
\end{itemize}
\end{corollary}

We now give some examples. First we extend notions of parametricity for categories to clans.

\begin{proposition}
\label{freeClanMonoidal}
The functor $F$ building free clans is monoidal.
\end{proposition}
\begin{proof}
This is exactly the same as Lemma \ref{freeLexMonoidal} for lex categories, using Lemma \ref{enrichedAdjunctionClans} instead of Lemma \ref{enrichedAdjunction}.
\end{proof}

\begin{corollary}
If $\C$ is a notion of parametricity for categories, then $F(\C)$ is a notion of parametricity for clans. The cubical model:
\begin{eqnarray}
F(\C)\multimap\D
\end{eqnarray}
is isomorphic to the clan where:
\begin{itemize}
\item Objects are functors from $\C$ to $\D$.
\item Morphisms are natural transformations.
\item Limits and fibrations are defined pointwise.
\end{itemize}
\end{corollary}
\begin{proof}
This is a reformulation of Lemma \ref{enrichedAdjunctionClans}.
\end{proof}

Now we introduce a genuinely new notion of parametricity.

\begin{definition}
Let $\square$ be a monoidal category with:
\begin{eqnarray}
\I&:&\square\\
d^0,d^1&:&\I\r 1
\end{eqnarray} 
such that any object in $\square$ is of the form $\I^{n}$ for some $n\geq 0$.

From this data we define a notion of parametricity for clans $\square_c$ as the free monoidal clan generated by:
\begin{itemize}
\item A monoidal functor from $\square$ to $\square_c$
\item The fact that we have a fibration:
\begin{eqnarray}
\langle d^0,d^1\rangle &:& \I\twoheadrightarrow 1\times 1
\end{eqnarray}
\end{itemize}
\end{definition}

\begin{definition}
A notion of parametricity for clans is called cubical if it comes from the previous definition.
\end{definition}

\begin{example}
The main example is the free monoidal clan generated by:
\begin{itemize}
\item An object $\I$.
\item A fibration $i:\I\twoheadrightarrow 1\times 1$.
\end{itemize}
It is called the clan of semi-cubes. This can be adapted to any variant of cubes.
\end{example}

\begin{remark}
The clan of semi-cubes corresponds to the usual parametricity for type theory, reformulated in the language of clans. Indeed the map:
\begin{eqnarray}
i &:& \I\r 1\times 1
\end{eqnarray}
is a fibration so that in any parametric clan $\C$ we have:
\begin{itemize}
\item For any object $\Gamma:\C$ we have a fibration:
\begin{eqnarray} 
a_*& \twoheadrightarrow &a\times a
\end{eqnarray}
Type-theoretically this means that any context $\Gamma\vdash$ comes with a relation:
\begin{eqnarray}
\Gamma,\Gamma\vdash\Gamma_*
\end{eqnarray}
\item For a fibration $a\twoheadrightarrow b$ in $\C$, we have a fibration:
\begin{eqnarray}
a_* &\twoheadrightarrow& b_*\times_{b\times b}(a\times a)
\end{eqnarray}
Type-theoretically this means that for $\Gamma\vdash A$ we have:
\begin{eqnarray}
\Gamma,\Gamma,\Gamma_*,A,A &\vdash& A_*
\end{eqnarray}
as is customary with parametricity.
\end{itemize} 
\end{remark}

We give some notations:

\begin{definition}
Given a clan $\C$ and a fibration $i : a\twoheadrightarrow b$ in $\C$, we define:
\begin{eqnarray}
a^n &:& \mathrm{Ob}_\C\\
\delta a^n &:& \mathrm{Ob}_\C\\
\delta i^n &:& a^n \twoheadrightarrow \delta a^n
\end{eqnarray}
by:
\begin{eqnarray}
a^0 &=& 1\\
a^{n+1} &=& a^n\otimes A
\end{eqnarray}
and:
\begin{eqnarray}
\delta a^0 &=& \top\\
\delta a^{n+1} &=& (\delta a^n)\otimes a \pullback{(\delta a^n)\otimes b} a^n\otimes b
\end{eqnarray}
with:
\begin{eqnarray}
\delta i^0 &=& \epsilon_1 : 1\r \top\\
\delta i^{n+1} &=& \delta i^n \odot i
\end{eqnarray}
where for:
\begin{eqnarray}
f &:& x\twoheadrightarrow x'\\
g &:& y\twoheadrightarrow y'
\end{eqnarray}
we write:
\begin{eqnarray}
f\odot g &:& x\otimes y \twoheadrightarrow (x\otimes y')\pullback{x'\otimes y'}(x'\otimes y)
\end{eqnarray}
\end{definition}

Now we can state the key lemma about cubical notions of parametricity. 

\begin{lemma}
\label{structureSquare}
Assume given a cubical notions of parametricity $\square_c$ built from $\square$, $\I$, $d^0$, $d^1$. The clan underlying $\square_c$ is the free clan generated by:
\begin{itemize}
\item A functor from the category $\square$ to $\square_c$.
\item The fact that $\delta i^n$ is a fibration for all $n\geq 0$, where $i=\langle d^0,d^1\rangle$.
\end{itemize}
\end{lemma}
\begin{proof}
We will show that both considered clans are isomorphic to an intermediate free clan $\square'$ generated by:
\begin{itemize}
\item A morphism of clan from the free monoidal clan generated by $\square$ to $\square'$.
\item The fact that $\delta i^n$ is a fibration in $\square'$ for all $n\geq 0$.
\end{itemize}

First we check that asking a morphism of clan from the free monoidal clan generated by the monoidal category $\square$ is the same as asking for a functor from $\square$. Indeed the following square commutes up to natural isomorphism:
\[\xymatrix{
\Cat\ar[d]_F & \Mon\Cat\ar[l]_{\overline{U}}\ar[d]^{\overline{F}}\\
\Clan & \Mon\Clan\ar[l]^{U}
}\]
because $F$ is strongly monoidal by Lemma \ref{freeClanMonoidal}. So we have natural isomorphisms:
\begin{eqnarray}
\Hom_\Clan(U\overline{F}(\square),\D) &\cong& \Hom_\Clan(F\overline{U}(\square),\D)\\
&\cong& \Hom_\Cat(\overline{U}(\square),\D)
\end{eqnarray} 
This means that the intermediate clan $\square'$ is equivalent to the free clan generated by a functor from $\square$ and $\delta i^n$ a fibration for all $n\geq 0$.

Next we prove that in the free monoidal clan generated by the monoidal category $\square$ the following conditions are equivalent:
\begin{itemize}
\item $\delta i^n$ is a fibration for all $n\geq 0$.
\item $i$ is a fibration and fibrations are stable by $\_\odot \_$.
\end{itemize}
This will implies that $\square_c$ is isomorphic to the intermediate clan $\square'$.

The fact that the second condition implies the first is clear, as $\delta i^0$ is always a fibration and:
\begin{eqnarray}
\delta i^{n+1} &=& \delta i^n\odot i
\end{eqnarray}

For the converse we assume that $\delta i^n$ is a fibration for all $n\geq 0$, and we want to check that fibrations are closed by $\_\odot\_$.  We prove inductively on the fibration $\psi$ that for any fibration $f:x\twoheadrightarrow y$, we have that $f\odot \psi$ is a fibration:

\begin{itemize}

\item For $g:a\r b$ and $g':b\r c$ with $\psi=g'\circ g$, we have a commutative square:
\[\xymatrix{
x\otimes a \ar[rrrr]^{f\odot(g'\circ g)}\ar[d]_{f\odot g} &&& & (y\otimes a)\pullback{y\otimes c}(x\otimes c) \ar[d]^\cong\\
(y\otimes a)\pullback{y\otimes b}(x\otimes b)\ar[rrrr]_{(y\otimes a)\pullback{y\otimes b} (f\odot g')} &&&& (y\otimes a)\pullback{y\otimes b}(y\otimes b)\pullback{y\otimes c}(x\otimes c)
}\]
so that if $f\odot g$ and $f\odot g'$ are fibrations, then so is $f\odot (g'\circ g)$.

\item For $h : a\r b$ and $g : c\r b$ with $\psi=\pi_a$we have a pullback square:
\[\xymatrix{
a\pullback{b}c \ar[d]_{\pi_a}\ar[r]& c\ar[d]^g\\
a\ar[r]_h& b
}\]
and a commutative square:
\[\xymatrix{
x\otimes(a\pullback{b}c) \ar[r]^{f\odot\pi_a}\ar[d]_\cong & (x\otimes a) \pullback{y\otimes a} (y\otimes(a\pullback{b}c))\ar[d]^\cong\ \\
(x\otimes a)\pullback{x\otimes b}(x\otimes c) \ar[d]_{(x\otimes a)\pullback{x\otimes b}(f\odot g)} &  (x\otimes a)\pullback{y\otimes a}(y\otimes a)\pullback{y\otimes b} y\otimes c \ar[d]^\cong\\
(x\otimes a)\pullback{x\otimes b}(x\otimes b) \pullback{y\otimes b}(y\otimes c)\ar[r]_\cong & (x\otimes a)\pullback{y\otimes b}(y\otimes c)\\
}\]
so that if $f\odot g$ is a fibration, then so is $f\odot \pi_a$.

\item For $a$ an object we have for $\epsilon_a:a\r \top$ a commutative square:
\[\xymatrix{
x\otimes a \ar[rrr]^{f\odot\epsilon_a}\ar[d]_{f\otimes a} && &(y\otimes a)\pullback{y\otimes\top}(x\otimes\top)\ar[d]^\cong\\
y\otimes a \ar[rrr]_\cong && &(y\otimes a)\pullback{\top}\top
}\]
and then we reason inductively on $a$, using the following isomorphisms:
\begin{eqnarray}
f\otimes 1 &\cong& f\\
f\otimes(a\otimes b) &\cong& (f\otimes a)\otimes b\\
f\otimes\top &\cong& \id_\top\\
f\otimes(a\pullback{b}c) &\cong& (f\otimes a)\pullback{f\otimes b}(f\otimes c)
\end{eqnarray}
Finally we have a commuting triangle:
\[\xymatrix{
x\otimes\I\ar[rd]^{f\odot i} \ar[rr]^{f\otimes\I} & & y\otimes \I\\
 & (y\otimes\I)\pullback{y\otimes(1\times 1)}(x\otimes(1\times 1)) \ar[ur]^{\pi_{y\otimes\I}}& \\
}\]
so that if $f\odot i$ is a fibration then so is $f\otimes \I$, and in turn so is $f\odot \epsilon_a$ for all object $a$.

\item We did not treat the case of $\_\odot\_$ and $i$. Using the fact that $\_\odot\_$ is associative, and iterating the previous isomorphisms (and their analogues in the left variable), it is enough to prove that a map built exclusively from $i$ and $\_\odot\_$ is a fibration. But such a map is isomorphic to $\delta i^n$ from some $n>0$, which is assumed to be a fibration.

\end{itemize}
\end{proof}

\begin{remark}
The only part of the proof relying crucially on that fact that we have a cubical notion of parametricity is the step where we prove without using stability of fibrations by $\_\odot\_$ that if $f$ a fibration, so is $f\otimes a$.
\end{remark}

Now we want to analyse cofreely parametric models for cubical notions of parametricity.

\begin{definition}
\label{defReedyFibrations}
Let $\C$ be a clan, and $\square_c$ be a cubical notion of parametricity generated by $\square$.
\begin{itemize}
\item A functor $X:\C^\square$ is called fibrant if for all $n\geq 0$ we have an induced fibration:
\begin{eqnarray}
X(\I^n)\twoheadrightarrow \overline{X}(\delta\I^n)
\end{eqnarray}
where $\overline{X}$ extends $X$ by commuting with limits.

\item A morphism between:
 \begin{eqnarray}
X,Y &:& \C^\square
\end{eqnarray}
with $X$ and $Y$ fibrant is called a fibration if it is a pointwise fibration, and for all $n\geq 0$ we have fibrations:
\begin{eqnarray}
X(\I^n)\twoheadrightarrow \overline{X}(\delta\I^n)\pullback{\overline{Y}(\delta\I^n)}Y(\I^n)
\end{eqnarray} 
where $\overline{X}$ (resp. $\overline{Y}$) extends $X$ (resp. $Y$) by commuting with limits.
\end{itemize}
\end{definition}

\begin{remark}
It should be noted that:
\begin{eqnarray}
\overline{X}(\delta\I^n)
\end{eqnarray}
can be defined inductively by commuting with limits only if the maps:
\begin{eqnarray}
X(\I^k) &\r& \overline{X}(\delta i^k)
\end{eqnarray} 
are fibrations for $k<n$.
\end{remark}

\begin{remark}
For $\square$ (the opposite of) a Reedy category of cubes where faces are the only dimension-decreasing maps, a fibration in the previous sense is precisely a Reedy fibration in $\C^\square$, and a fibrant object in $\C^\square$ is precisely a Reedy fibrant object. This is the case for semi-cubes, or cubes with reflexivities only. On the other hand diagonals contradict this assumption.
\end{remark}

\begin{proposition}
Let $\C$ be a clan, and $\square_c$ be a cubical notion of parametricity generated by $\square$. Then:
\begin{eqnarray}
\square_c\multimap\C
\end{eqnarray}
is the clan of fibrant objects in $\C^\square$, equipped the fibrations from Definition \ref{defReedyFibrations}.
\end{proposition}

\begin{proof}
By Lemma \ref{structureSquare} giving an object in $\square_c\multimap\C$ (i.e. a morphism of clan from $\square_c$ to $\C$) is the same as giving a fibrant object in $\C^\square$.


Next we want to extend this to show that a fibration in $\square_c\multimap\C$ is a fibration in the sense of Definition \ref{defReedyFibrations}. We define $\Delta$ as the free clan generated by two objects and a fibration between them. For $\C$ a clan we have that $\Delta\multimap\D$ is the clan such that:
\begin{itemize}
\item Its objects are fibrations in $\D$.
\item Its morphisms are commutative squares.
\item Its fibrations are Reedy fibrant squares.
\end{itemize}

Then we have an isomorphism:
\begin{eqnarray}
\Hom_\Clan(\Delta,\square_c\multimap\C) &\cong& \Hom_\Clan(\square_c,\Delta\multimap\C)
\end{eqnarray}
so that, by Lemma \ref{structureSquare}, a fibration in $\square_c\multimap\C$ between $X$ and $Y$ is the same thing as:
\begin{itemize}
\item A functor from $\square$ to $\Delta\multimap\C$ with source $X$ and target $Y$. This is equivalent to a natural transformation $\alpha$ from $X$ to $Y$ restricted to $\square$.
\item Such that for any $n\geq 0$ we have that:
\[\xymatrix{
X(\I^n) \ar[r]^{\alpha(\I^n)}\ar[d]& \ar[d]Y(\I^n)\\
X(\delta\I^n)\ar[r]_{\overline{\alpha}(\delta\I^n)} & Y(\delta\I^n)
}\]
is Reedy fibrant, where $\overline{\alpha}$ extends $\alpha$ by commuting with limits. 
\end{itemize}
But this is precisely a fibration in the sense of Definition \ref{defReedyFibrations}.
\end{proof}

\begin{remark}
This allows to generate clans of Reedy fibrant semi-cubical and cubical objects. We conjecture that this result can be extended to all Reedy category with a compatible monoidal structure, using the monoidal clan $\square$ with dimension-decreasing maps as fibrations. It is unclear to us what \emph{compatible} should mean here. 
\end{remark}






\section*{Conclusion}

First we give a summary of the presented a result, and then possible directions for future work.

\subsection*{Summary}

We showed that notions of parametricity can be defined as monoidal models of type theory, and parametric models as modules. In one diagram:
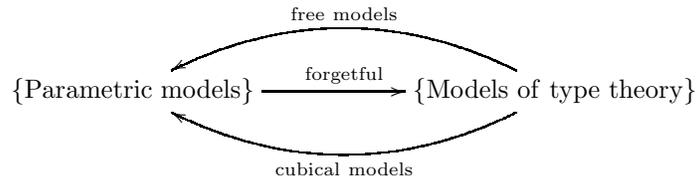
\[\xymatrix{
\{\mathrm{Parametric\ models}\}\ar[rr]^{\mathrm{forgetful}} & & \ar@/^-2pc/[ll]_{\mathrm{free\ models}}\ar@/^2pc/[ll]^{\mathrm{cubical\ models}}\{\mathrm{Models\ of\ type\ theory}\}
}\]
We elaborate a bit:

\begin{itemize}
\item Assuming a symmetric monoidal closed category of models of type theory, a notion of parametricity was defined as a monoid $\square$ (i.e. a monoidal model of type theory). Then a parametric model was defined as a model of type theory with a $\square$-action, i.e. a $\square$-module. 
\item We gave explicit formulas for the left and right adjoints to the functor forgetting a module structure, yielding freely parametric models and cubical models.
\item We gave symmetric monoidal closed structures on the categories of categories, lex categories and clans, all of them seen as models of type theory.
\item We gave many examples of parametricity, explaining how to construct the following as cofreely parametric models of type theory:
\begin{itemize}
\item Categories of cubical objects (for any variant of cubes).
\item Lex categories of truncated cubical objects (for semi-cubes and cubes with reflexivities).
\item Clans of Reedy fibrant cubical objects (for semi-cubes and cubes with reflexivities).
\end{itemize}
\end{itemize}

\subsection*{Further work}

We list several limitations of this work, and offer directions to deal with them.

\begin{itemize}
\item An obvious technical limitation is the requirement of strict commutativity of limits in lex categories and clans. The proper way to remove this is to work with a $2$-category of models of type theory, removing the need for any strictness assumption. 

An alternative, less satisfying approach would be to extend well-known strictification results for monoidal categories to monoidal lex categories and monoidal clans.
\item We should construct symmetric monoidal closed structures on many other categories of models of type theory, including categories with families and comprehension categories. We want to design techniques allowing to do this efficiently.

We conjecture that all these monoidal closed structure can be build as some variant of Day convolution, using the fact that the lex categories classifying models of type theory are monoidal (for reasonable notions of model of type theory with $\Sigma$ and $\top$). This would imply that any model of type theory is a model internal to the category of models of type theory. This is analogous to the fact that any abelian group is an abelian group internal to the category of abelian groups.
\item It is unclear how to include $\Pi$ and $\U$ in this context. Indeed even in the relatively simple case of categories we cannot extend parametricity to exponentials. This would imply that the counit sending any cubical object to its underlying object of points commutes with exponentials, which fails in the presence of reflexivities (although this is true for semi-cubes, as seen in \cite{moeneclaey2021parametricity}).

However one can prove that when a category $\D$ has exponentials \emph{and enough limits}, the category $\D^\C$ has exponentials for any category $\C$. Perhaps this can be extended to a result giving $\Pi$ and $\U$ in any cubical model, assuming $\Pi$ and $\U$ together with \emph{enough inductive types} in the base model.
\item We cannot build clans with Kan fibrations as cofreely parametric models. Indeed for any clans $\square$ and $\D$, the fibrations in $\square\multimap\D$ are always defined by the requirement that some maps in $\D$ are fibrations, and never by the requirement that some terms exist, but these terms seem necessary to model the lifting in a Kan fibration. 

Nevertheless we think that Kan cubical models can be constructed in two steps: first construct cubical objects, then restrict the fibrations. We hope that the second step can be compactly formulated in our setting, even if it does not consist in adding a module structure. In particular we conjecture that there is a right adjoint to the forgetful functor:
\begin{eqnarray}
\{\mathrm{Parametric\ models\ with\ Kan\ liftings}\}&\r& \{\mathrm{Parametric\ models}\}\nonumber
\end{eqnarray}
The notion of \emph{parametricity with Kan liftings} should be equivalent to univalence, making formal the following correspondence:
\[\begin{tabular}{|c|c|c|}
\hline
Relation & Parametricity & Semi-cubes\\
Equivalence & Univalence & Kan cubes\\
\hline
\end{tabular}\]
This idea was already explored by Altenkirch and Kaposi \cite{altenkirch2017towards} without fully computing with univalence, and by Tabareau, Tanter and Sozeau \cite{tabareau2021marriage} assuming a univalent universe to begin with.
\end{itemize}

\bibliographystyle{plain}
\bibliography{parametricity_monoidal}

\end{document}